\documentclass[a4paper]{article}

\usepackage[left=3cm,right=3cm,top=3cm,bottom=3cm]{geometry}
\usepackage{authblk}

\usepackage{amsmath}
\usepackage{amssymb}
\usepackage{amsfonts}
\usepackage{amsthm}
\usepackage{adjustbox}
\usepackage{caption}
\usepackage{graphicx}
\usepackage{hyperref}
\usepackage{cleveref}

\usepackage{colonequals}

\usepackage{algorithm}
\usepackage{algpseudocode}

\newcommand*{\logeq}{%
	\mathrel{\ratio\colonsep\Longleftrightarrow}%
}
\newcommand{\sqsubsetdot}{\mathrel{\ooalign{$\sqsubset$\cr
  \hidewidth\hbox{$\cdot\mkern2mu$}\cr}}}

\usepackage{todonotes}

\newtheorem{lemma}{Lemma}[section]
\newtheorem{theorem}[lemma]{Theorem}

\theoremstyle{definition}
\newtheorem{definition}[lemma]{Definition}

\newtheorem{remark}[lemma]{Remark}
\newtheorem{hyp}[lemma]{Hypothesis}

\newcommand\restr[2]{{
  \left.\kern-\nulldelimiterspace 
  #1 
  \right|_{#2} 
  }}

\DeclareMathOperator{\subseq}{\preceq}

\DeclareMathOperator{\DFA}{DFA}

\DeclareMathOperator{\states}{\mathsf{states}}

\DeclareMathOperator{\bigO}{O}
\DeclareMathOperator{\bigo}{O}

\DeclareMathOperator{\npclass}{NP}
\DeclareMathOperator{\pclass}{P}

\DeclareMathOperator{\wclass}{W}

\DeclareMathOperator{\reg}{REG}
\DeclareMathOperator{\semi}{SLS}

\DeclareMathOperator{\emptyword}{\varepsilon}

\DeclareMathOperator{\REG}{REG}
\DeclareMathOperator{\poly}{poly}

\DeclareMathOperator{\vsn}{vsn}
\DeclareMathOperator{\interval}{interval}
\DeclareMathOperator{\lab}{label}

\DeclareMathOperator{\PatClass}{\mathcal{P}}

\newcommand{\simpleLSConstAlph}[1]{\PatClass^{k}_{\mathsf{s}, \mathsf{LC}}}

\DeclareMathOperator{\OV}{\textsf{OV}}

\newcommand{\gapsize}[1]{\mathsf{gapsize}(#1)}

\newcommand{\lang}[1]{\mathcal{L}(#1)}

\newcommand{\gapName}{\mathsf{gap}}
\newcommand{\gap}{\gapName}

\DeclareMathOperator{\size}{\mathsf{size}}

\newcommand{\matchProb}{\textsc{Match}}
\newcommand{\cliqueProb}{\textsc{Clique}}

\DeclareMathOperator{\ta}{\mathtt{a}}
\DeclareMathOperator{\tb}{\mathtt{b}}
\DeclareMathOperator{\tc}{\mathtt{c}}

\theoremstyle{plain}

\begin{document}

\title{Subsequences With Generalised Gap Constraints: Upper and Lower Complexity Bounds}

\author[1]{Florin Manea}
\author[1]{Jonas Richardsen}
\author[2]{Markus L.\ Schmid}

\affil[1]{Computer Science Department and CIDAS, Universität Göttingen, Germany, \texttt{j.richardsen@stud.uni-goettingen.de}, \texttt{florin.manea@cs.informatik.uni-goettingen.de}}
\affil[2]{Humboldt-Universit\"at zu Berlin, Germany, \texttt{MLSchmid@MLSchmid.de}}

\date{}

\maketitle

\begin{abstract}
For two strings $u, v$ over some alphabet $A$, we investigate the problem of embedding $u$ into $w$ as a subsequence under the presence of generalised gap constraints. A generalised gap constraint is a triple $(i, j, C_{i, j})$, where $1 \leq i < j \leq |u|$ and $C_{i, j} \subseteq A^*$. Embedding $u$ as a subsequence into $v$ such that $(i, j, C_{i, j})$ is satisfied means that if $u[i]$ and $u[j]$ are mapped to $v[k]$ and $v[\ell]$, respectively, then the induced gap $v[k + 1..\ell - 1]$ must be a string from $C_{i, j}$. This generalises the setting recently investigated in [Day et al., ISAAC 2022], where only gap constraints of the form $C_{i, i + 1}$ are considered, as well as the setting from [Kosche et al., RP 2022], where only gap constraints of the form $C_{1, |u|}$ are considered. 

We show that subsequence matching under generalised gap constraints is NP-hard, and we complement this general lower bound with a thorough (parameterised) complexity analysis. Moreover, we identify several efficiently solvable subclasses that result from restricting the interval structure induced by the generalised gap constraints.
\end{abstract}

\section{Introduction}

For a string $v = v_1 v_2 \ldots v_n$, where each $v_i$ is a single symbol from some alphabet $\Sigma$, any string $u = v_{i_1} v_{i_2} \ldots v_{i_k}$ with $k \leq n$ and $1 \leq i_1 < i_2 < \ldots < i_{k} \leq n$ is called a \emph{subsequence} (or \emph{scattered factor} or \emph{subword}) of $v$ (denoted by $u \subseq v$). This is formalised by the \emph{embedding} from the positions of $u$ to the positions of $v$, i.\,e., the increasing mapping $e : \{1, 2, \ldots, k\} \to \{1, 2, \ldots, n\}$ with $j \mapsto i_j$ (we use the notation $u \subseq_e v$ to denote that $u$ is a subsequence of $v$ via embedding~$e$). For example, the string $\ta \tb \ta \tc \tb \tb \ta$ has among its subsequences $\ta \ta \ta$, $\ta \tb \tc \ta$, $\tc \tb \ta$, and $\ta \tb \ta \tb \tb \ta$. With respect to $\ta \ta \ta$, there exists just one embedding, namely $1 \mapsto 1$, $2 \mapsto 3$, and $3 \mapsto 7$, but there are two embeddings for $\tc \tb \ta$. \par

This classical concept of subsequences is employed in many different areas of computer science: in formal languages and logics (e.\,g., piecewise testable languages~\cite{simonPhD,Simon72,KarandikarKS15,CSLKarandikarS,journals/lmcs/KarandikarS19,PraveenEtAl2024}, or subword order and downward closures~\cite{HalfonSZ17,KuskeZ19,Kuske20,Zetzsche16}), in combinatorics on words~\cite{RigoS15,FreydenbergerGK15,LeroyRS17a,Rigo19,Seki12,Mat04,Salomaa05,SchnoebelenVeron2023}, for modelling concurrency~\cite{Riddle1979a, Shaw1978, BussSoltys2014}, in database theory (especially event stream processing~\cite{ArtikisEtAl2017,GiatrakosEtAl2020,ZhangEtAl2014,KleestMeissnerEtAl2021,Kleest-MeissnerEtAl23,FrochauxKleestMeissner2023}). Moreover, many classical algorithmic problems are based on subsequences, e.\,g., {longest common subsequence} \cite{DBLP:journals/tcs/Baeza-Yates91} or {shortest common supersequence} \cite{Maier:1978} (see~\cite{AdamsonEtAl2023, FleischmannEtAl2023} and the survey \cite{surveyNCMA}, for recent results on string problems concerned with subsequences). Note that the longest common subsequence problem, in particular, has recently regained substantial interest in the context of fine-grained complexity (see~\cite{DBLP:conf/fsttcs/BringmannC18,BringmannK18,AbboudEtAl2015,AbboudEtAl2014}).\looseness=-1\par

In this paper, we are concerned with the following special setting of subsequences recently introduced in~\cite{DayEtAl2022}. If a string $u$ is a subsequence of a string $v$ via an embedding $e$, then this embedding $e$ also induces $|u|-1$ so-called \emph{gaps}, i.\,e., the (possibly empty) factors $v_{e(i) + 1} v_{e(i) + 2} \ldots v_{e(i + 1) - 1}$ of $v$ that lie strictly between the symbols where $u$ is mapped to. For example, $\ta \tc \tb \subseq_e \ta \tb \ta \tc \tb \tb \ta$ with $e$ being defined by $1 \mapsto 1$, $2 \mapsto 4$, and $3 \mapsto 6$ induces the gaps $\tb \ta$ and $\tb$. We can now restrict the subsequence relation by adding \emph{gap constraints} as follows. A string $u$ is accompanied by $|u|-1$ gap constraints $C_1, C_2, \ldots, C_{|u|-1} \subseteq \Sigma^*$, and $u$ is a valid subsequence of a string $v$ under these gap constraints, if $u \subseq_e v$ for an embedding $e$ that induces gaps from the gap constraints, i.\,e., the $i^{\text{th}}$ gap is in $C_i$. 

Such gap-constrained subsequences allow to model situations for which classical subsequences are not expressive enough. For example, if we model concurrency by shuffling together strings that represent threads on a single processor, then fairness properties of a scheduler usually imply that the gaps of these subsequences are not huge. Or assume that we compute an alignment between two strings by computing a long common subsequence. Then it is not desirable if roughly half of the positions of the common subsequence are mapped to the beginning of the strings, while the other half is mapped to the end of the strings, with a huge gap (say thousands of symbols) in between. In fact, an overall shorter common subsequence that does not contain such huge gaps seems to induce a more reasonable alignment (this setting is investigated in~\cite{AdamsonEtAl2023}). Another example is complex event processing: Assume that a log-file contains a sequence of events of the run of a large system. Then we might query this string for the situation that between some events of a job $A$ only events associated to a job $B$ appear (e.\,g., due to unknown side-effects this leads to a failure of job $A$). This can be modeled by embedding a string as a subsequence such that the gaps only contain symbols from a certain subset of the alphabet, i.\,e., the events associated to job $B$ (such subsequence queries are investigated in~\cite{KleestMeissnerEtAl2021,Kleest-MeissnerEtAl23,FrochauxKleestMeissner2023}).

In~\cite{DayEtAl2022}, two types of gap constraints are considered: Length constraints $C = \{w \in \Sigma^* \mid \ell \leq |w| \leq k\}$, and regular constraints where $C$ is just a regular language over $\Sigma^*$, as well as combinations of both. In a related paper, \cite{BoundedSubsequences}, the authors went in a slightly different direction, and were interested in subsequences appearing in bounded ranges, which is equivalent to constraining the length of the string occurring between the first and last symbol of the embedding, namely $v_{e(1) + 1} v_{e(i) + 2} \ldots v_{e(m) - 1}$. In this paper, we follow up on the work of~\cite{DayEtAl2022,BoundedSubsequences}, but significantly generalise the concept of gap constraints. Assume that $u \subseq_e v$. Instead of only considering the gaps given by the images of two consecutive positions of $u$, we consider each string $v_{e(i) + 1} v_{e(i) + 2} \ldots v_{e(j) - 1}$ of $v$ as a gap, where $i, j \in \{1, 2, \ldots, |u|\}$ with $i < j$ (note that these general gaps also might contain symbols from $v$ that correspond to images of $e$, namely $e(i + 1), e(i + 2), \ldots, e(j-1)$). For example, $\ta \tb \ta \tc \subseq_e \tb \ta \ta \tb \tb \tc \ta \tc \tc \ta \tb$ with $e$ defined by $1 \mapsto 2$, $2 \mapsto 5$, $3 \mapsto 7$ and $4 \mapsto 9$ induces the following gaps: The $(1, 2)$-gap $\ta \tb$, the $(2, 3)$-gap $\tc$, the $(3, 4)$-gap $\tc$, the $(1, 3)$-gap $\ta \tb \tb \tc$, the $(2, 4)$-gap $\tc \ta \tc$, and the $(1, 4)$-gap $\ta \tb \tb \tc \ta \tc$. In this more general setting, we can now add gap-constraints in an analogous way as before. For example, the gap constraint $C_{2, 4} = \{\ta, \tc\}^*$ for the $(2, 4)$-gap, the gap constraint $C_{1, 4} = \{w \in \Sigma^* \mid 3 \leq |w| \leq 5\}$ for the $(1, 4)$-gap and the gap constraint $C_{2, 3} = \{\tc^n \mid n \geq 1\}$ for the $(2, 3)$-gap. Under these gap-constraints, the embedding $e$ defined above is not valid: The gap constraints $C_{2, 4}$ and $C_{2, 3}$ are satisfied, but the $(1, 4)$-gap $\ta \tb \tb \tc \ta \tc$ is too long for gap constraint $C_{1, 4}$. However, changing $4 \mapsto 9$ into $4 \mapsto 8$ yields an embedding that satisfies all gap constraints. \looseness=-1

 \textbf{Our Contribution.} We provide an in-depth analysis of the complexity of the \emph{matching problem} associated with the setting explained above, i.\,e., for given strings $u, v$ and a set $\mathcal{C}$ of generalised gap-constraints for $u$, decide whether or not $u \subseq_e v$ for an embedding $e$ that satisfies all constraints in $\mathcal{C}$. We concentrate on two different kinds of constraints: semilinear constraints of the form $\{w \in \Sigma^* \mid |w| \in S\}$, where $S$ is a semilinear set, and regular constraints.\looseness=-1

In general, this matching problem is $\npclass$-complete for both types of constraints (demonstrating a stark contrast to the simpler setting of gap constraints investigated in~\cite{DayEtAl2022,BoundedSubsequences}), and this even holds for binary alphabets and if each semilinear constraint has constant size, and also if every regular constraint is represented by an automaton with a constant number of states. On the other hand, if the number of constraints is bounded by a constant, then the matching problem is solvable in polynomial-time, but, unfortunately, we obtain $\wclass[1]$-hardness even if the complete size $|u|$ is a parameter (also for both types of constraints).  An interesting difference in complexity between the two types of constraints is pointed out by the fact that for regular constraints the matching problem is fixed-parameter tractable if parameterised by $|u|$ and the maximum size of the regular constraints (measured in the size of a DFA), while for semilinear constraints this variant stays $\wclass[1]$-hard. 

We then show that structurally restricting the interval structure induced by the given constraints yields polynomial-time solvable subclasses. Moreover, if the interval structure is completely non-intersecting, then we obtain an interesting subcase for which the matching problem can be solved in time $O(n^\omega |\mathcal{C}|)$, where $O(n^\omega)$ is the time needed to multiply two  $n\times n$ Boolean matrices. We complement this result by showing that an algorithm with running time $\mathcal{O}(|w|^g|\mathcal{C}|^h)$ with $g + h < 3$ would refute the strong exponential time hypothesis. While this is not a tight lower bound, we wish to point out that, due to the form of our algorithm, which boils down to performing $O(|{\mathcal C}|)$ matrix multiplications, a polynomially stronger lower bound would have proven that matrix multiplication in quadratic time is not possible.\looseness=-1

\textbf{Related Work.} Our work extends~\cite{DayEtAl2022,BoundedSubsequences}. However, subsequences with various types of gap constraints have been considered before, mainly in the field of combinatorial pattern matching with biological motivations (see~\cite{BilleEtAl2012, LiW08,LiYWL12, IliopoulosEtAl2007} and~\cite{BaderEtAl2016,CaceresEtAl2020} for more practical papers).

\section{Preliminaries}

Let $\mathbb{N} = \{1, 2, \dots\}$, $\mathbb{N}_0 = \mathbb{N} \cup \{0\}$. For $m, n \in \mathbb{N}_0$ let $[m, n] = \{k \in \mathbb{N}_0 \mid m \leq k \leq n\} = \{m, \dots, n\}$ and $[n] = [1, n]$. For some alphabet $\Sigma$ and some length $n \in \mathbb{N}_0$ we define $\Sigma^n$ as the set of all words of length $n$ over $\Sigma$ (with $\Sigma^0$ only containing the empty word $\varepsilon$). Furthermore $\Sigma^* \colonequals \bigcup_{n \in \mathbb{N}_0} \Sigma^n$ is the set of all words over $\Sigma$. For some $w \in \Sigma^*$, $|w|$ is the length of $w$, $w[i]$ denotes the $i$-th character of $w$ and $w[i..j] \colonequals w[i] \dots w[j]$ is the substring of $w$ from the $i$-th to the $j$-th character (where $i, j \in [|w|], i \leq j$).

We use deterministic and nondeterministic finite automata (DFA and NFA) as commonly defined in the literature; as a particularity, for the sake of having succinct representations of automata, we allow  DFAs to be incomplete: given a state $q$ of a DFA and a letter $a$, the transition from $q$ with $a$ may be left undefined, which means that the computations of the DFA on the inputs which lead to the respective transition are not-accepting. For a DFA or NFA $A$, we denote by $\size(A)$ its total size, and by $\states(A)$ its number of states. Note that if $A$ is a DFA over alphabet $\Sigma$, then we have that $\size(A) = \bigO(\states(A) |\Sigma|)$.

A subset $L \subseteq \mathbb{N}$ is called \emph{linear}, if there are $m \in \mathbb{N}_0$ and $x_0 \in \mathbb{N}_0$, $x_1, \dots, x_m \in \mathbb{N}$, such that $ L = L(x_0; x_1, \dots, x_m) \colonequals \left\{ x_0 + \sum_{i=1}^m k_ix_i \,\middle\vert\, k_1, \dots, k_m \in \mathbb{N}_0\right\}\,.$
For $m = 0$, we write $L(x_0) = \{x_0\}$. We can assume without loss of generality that $x_i\neq x_j$ for $i\neq j, i,j\in [m]$. A set $S$ is \emph{semilinear}, if it is a finite union of linear sets (see also \cite{SemilinearSets}).

We assume that each integer involved in the representation of a linear set fits into constant memory (see our discussion about the computational model at the end of this section). Consequently, we measure the size of a linear set $L = L(x_0; x_1, \dots, x_m)$ as $\size(L) = m + 1$, and the size of a semilinear set $S = L_1 \cup L_2 \cup \ldots \cup L_k$ is measured as $\size(S) = \sum^k_{i = 1} \size(L_i)$. In other words, $\size(S)$ is the number of integers used for defining $S$.

\medskip

\noindent\textbf{Computational Model}:
For the complexity analysis of the algorithmic problems described in this paper we assume the \emph{unit-cost RAM model with logarithmic word size} (see \cite{cormen}). This means that for input size $N$, the memory words of the model can store $\log N$ bits. Thus, if we have input words of length $N$, they are over an alphabet which has at most $\sigma \leq N$ different characters, which we can represent using the \emph{integer alphabet} $\Sigma = [\sigma]$. As such, we can store each character within one word of the model. Then, it is possible to read, write and compare single characters in one unit time. 

\medskip

\noindent\textbf{Complexity Hypotheses}: Let us consider the \emph{Satisfiability problem for formulas in conjunctive normal form}, or \textsc{CNF-Sat} for short. Here, given a boolean formula $F$ in conjunctive normal form, i.e., $F = \{c_1, \dots, c_m\}$ and $c_i \subseteq \{v_1, \dots, v_n, \neg v_1, \dots, \neg v_n\}$ for variables $v_1, \dots, v_n$, it is to be determined whether $F$ is satisfiable. This problem was shown to be NP-hard \cite{CNFSAT}. By restricting $|c_i| \leq k$ for all $i \in [m]$ we obtain the problem of \textsc{$k$-CNF-Sat}. We will base our lower bound on the following algorithmic hypothesis:

\begin{hyp}[Strong Exponential Time Hypothesis (\textsf{SETH}) \cite{StrongETH}]
    \label{hyp_SETH}
    For any $\varepsilon > 0$, there exists a $k \in \mathbb{N}$, such that \textsc{$k$-CNF-Sat} cannot be solved in $\mathcal{O}(2^{n(1-\varepsilon)} \operatorname{poly}(m))$ time, where $\operatorname{poly}(n)$ is an arbitrary (but fixed) polynomial function.
\end{hyp}

The \emph{Clique problem}, \textsc{Clique}, asks, given a graph $G$ and a number $k \in \mathbb{N}$, whether $G$ has a $k$-clique. Hereby, a $k$-clique is a subset of $k$ pairwise adjacent vertices, i.e., there is an edge between any pair of vertices in the subset. Since \textsc{CNF-Sat} can be reduced to \textsc{Clique} \cite{KClique}, the latter is also NP-hard.

The \emph{$k$-Orthogonal Vectors problem}, $k$-$\OV$, receives as inputs $k$ sets $V_1, \dots, V_k$ each containing $n$ elements from $\{0, 1\}^d$ for some $d \in \omega(\log n)$, i.e., $d$-dimensional boolean vectors. The question is, whether one can select vectors $\vec{v}_i \in V_i$ for $i \in [k]$ such that the vectors are orthogonal: $\sum_{j = 1}^n \prod_{i = 1}^k \vec{v}_i[j] = 0$.
It is possible to show the following lemma (\cite{kOV,kOVBase}):

\begin{lemma}
    \label{lem_k_OV}
	$k$-$\OV$ cannot be solved in $n^{k-\varepsilon} \operatorname{poly}(d)$ time for any $\varepsilon > 0$, unless \textsf{SETH} fails.\looseness=-1
\end{lemma}

This lemma will later form the basis for the conditional lower bound in the case of non-intersecting constraints.

\section{Subsequences with Gap Constraints}\label{sec:gapConstraints} 

An \emph{embedding} is any function $e : [k] \to [\ell]$ for some $k, \ell \in \mathbb{N}$ with $k \leq \ell$, such that $e(1) < e(2) < \ldots < e(k)$ (note that this also implies that $1 \leq e(1)$ and $e(k) \leq \ell$). Let $\Sigma$ be some alphabet. For a string $v = v_1 v_2 \ldots v_n$, where $v_i \in \Sigma$ for every $i \in [n]$, any string $u = v_{i_1} v_{i_2} \ldots v_{i_k}$ with $k \leq n$ and $1 \leq i_1 < i_2 < \ldots < i_{k} \leq n$ is called a \emph{subsequence} (or, altternatively, \emph{scattered factor} or \emph{subword}) of $v$ (denoted by $u \subseq v$). Every embedding $e : [k] \to [|v|]$ with $k \leq |v|$ \emph{induces} the subsequence $u_e = v_{e(1)} v_{e(2)} \ldots v_{e(k)}$ of $v$. If $u$ is a subsequence of $v$ induced by an embedding $e$, then we denote this by $u \subseq_e v$; we also say that an embedding $e$ \emph{witnesses} $u \subseq v$ if $u \subseq_e v$. When embedding a substring $u[s..t]$ for some $s, t \in [|u|], s \leq t$, we use a partial embedding $e: [s, t] \to [n]$ and write $u[s..t] \preceq_e w$.\par
For example,
the string $\ta \tb \ta \tc \tb \tb \ta$ has among its subsequences $\ta \ta \ta$, $\ta \tb \tc \ta$, $\tc \tb \ta$, and $\ta \tb \ta \tb \tb \ta$. With respect to $\ta \ta \ta$, there exists just one embedding, namely $1 \mapsto 1$, $2 \mapsto 3$, and $3 \mapsto 7$, but there are two different embeddings for $\tc \tb \ta$. \par
Let $v \in \Sigma^*$ and let $e : [k] \to [|v|]$ be an embedding. For every $i, j \in [k]$ with $i < j$, the string $v$ and the embedding $e$ induces the $(i, j)$-\emph{gap}, which is the factor of $v$ that occurs strictly between the positions corresponding to the images of $i$ and $j$ under the embedding $e$, i.\,e., $\gap_{v, e}[i, j] = v[e(i)+1..e(j)-1]$. If $v$ and $e$ are clear from the context, we also drop this dependency in our notation, i.\,e., we also write $\gap[i, j]$.\par
As an example, consider $v = \ta \tb \tc \tb \tc \ta \tb \tc \ta \tb \ta \tc$ and $u = \ta \tc \ta \tb \ta$. There are several embeddings $e : [|u|] \to [|v|]$ that witness $u \subseq v$. Each such embedding also induces an $(i, j)$-gap for every $i, j \in [5]$ with $i < j$. For the embedding $e$ with $e(1) = 1$, $e(2) = 3$, $e(3) = 6$, $e(4) = 7$, $e(5) = 11$ (that satisfies $u \subseq_e v$), some of these gaps are illustrated below (note that $e$ is also indicated by the boxed symbols of $v$):
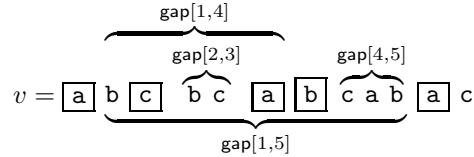
\begin{equation*}
v =  \boxed{\ta}\:\underbrace{\overbrace{\tb\:\boxed{\tc}\:\overbrace{\tb\:\tc}^{\gap[2, 3]}\:\boxed{\ta}}^{\gap[1, 4]}\:\boxed{\tb}\:\overbrace{\tc\:\ta\:\tb}^{\gap[4, 5]}}_{\gap[1, 5]}\:\boxed{\ta}\:\tc
\end{equation*}

A \emph{gap-constraint} for a string $u \in \Sigma^*$ is a triple $C = (i, j, L)$ with $1 \leq i < j \leq |u|$ and $L \subseteq \Sigma^*$. A gap-constraint $C = (i, j, L)$ is also called an \emph{$(i, j)$-gap-constraint}. The component $L$ is also called the \emph{gap-constraint language} of the gap-constraint $(i, j, L)$.
We say that a string $v$ and some embedding $e : [|u|] \to [|v|]$ satisfies the gap-constraint $C$ if and only if $\gap_{v, e}[i, j] \in L$. 

As an example, let us define some gap-constraints for the string $u = \ta \tc \ta \tb \ta$: $(1, 4, \Sigma^*)$, $(1, 5, \{w_1 \tc w_2 \tc w_3 \mid w_1, w_2, w_3  \in \Sigma^*\})$, $(2, 3, \{w \in \Sigma^* \mid |w| \geq 5\})$ and $(4, 5, \{w \in \Sigma^* \mid |w| \leq 4\})$. It can be easily verified that the gaps induced by the string $v$ and the embedding $e$ defined above (and illustrated by the figure) satisfy all of these gap constraints, except $(2, 3, \{w \in \Sigma^* \mid |w| \geq 5\})$ since $|\gap_{v, e}[2, 3]| = 2 < 5$. However, the embedding $e'$ defined by $e'(1) = 1$, $e'(2) = 3$, $e'(3) = 9$, $e'(4) = 10$ and $e'(5) = 11$  is such that $v$ and $e'$ satisfy all of the mentioned gap-constraints (in particular, note that $|\gap_{v, e'}[2, 3]| = 5$ and that $\gap_{v, e'}[4, 5] = \emptyword$).

A \emph{set of gap-constraints} for $u$ is a set $\mathcal{C}$ that, for every $i, j \in \mathbb{N}$ with $i < j$, may contain at most one $(i, j)$-gap-constraint for $u$. A string $v$ and some embedding $e : [|u|] \to [|v|]$ satisfy $\mathcal{C}$ if $v$ and $e$ satisfy every gap-constraint of $\mathcal{C}$.\par
For strings $u, v \in \Sigma^*$ with $|u| \leq |v|$, and a set $\mathcal{C}$ of gap-constraints for $u$, we say that $u$ is a \emph{$\mathcal{C}$-subsequence of $v$}, if $u \subseq_e v$ for some embedding $e$ such that $v$ and $e$ satisfy $\mathcal{C}$. We shall also write $u \subseq_{\mathcal{C}} v$ to denote that $u$ is a $\mathcal{C}$-subsequence of $v$.

\medskip

\noindent\textbf{The Matching Problem}: The central decision problem that we investigate in this work is the following matching problem, $\matchProb$, for subsequences with gap-constraints:
\begin{itemize}
	\item \emph{Input:} Two strings $w \in \Sigma^*$ (also called text), with $|w|=n$, and $p \in \Sigma^*$ (also called pattern), with $|p|=m\leq n$,  and a non-empty set $\mathcal{C}$ of gap-constraints.
	\item \emph{Question:} Is $p$ a $\mathcal{C}$-subsequence of $w$?
\end{itemize}

Obviously, for this matching problem it is vital how we represent gap-constraints $(i, j, L)$, especially the gap-constraint language $L$. Moreover, since every possible language membership problem ``$v \in L?$'' can be expressed as the matching problem instance $w = \# v \#$, $p = \#\#$ and $\mathcal{C} = (1, 2, L)$, where $\# \notin \Sigma$, we clearly should restrict our setting to constraints $(i, j, L)$ with sufficiently simple languages $L$. These issues are discussed next.

\medskip

\noindent\textbf{Types of Gap-Constraints}: A gap-constraint $C = (i, j, L)$ (for some string) is a

\begin{itemize}
\item \emph{regular constraint} if $L \in \REG$. We represent the gap-constraint language of a regular constraint by a deterministic finite automaton (for short, $\DFA$). In particular, $\size(C) = \size(L) = \size(A)$ and $\states(C) = \states(L) = \states(A)$.
\item \emph{semilinear length constraint} if there is a semi-linear set $S$, such that $L = \{w \in \Sigma^* \mid |w| \in S\}$. We represent the gap-constraint language of a semi-linear length constraint succinctly by representing the semilinear set $S$ in a concise way (i.\,e., as numbers in binary encoding). In particular, $\size(C) = \size(L) = \size(S)$.
\end{itemize}
For a set $\mathcal{C}$ of gap constraints, let $\size(\mathcal{C}) = \sum_{C \in \mathcal{C}} \size(C)$ and $\gapsize{\mathcal{C}} = \max\{\size(C) \mid C \in \mathcal{C}\}$. We have $\size({\mathcal C})\leq |{\mathcal C}|\gapsize{\mathcal C}$.

Obviously, $\{v \in \Sigma^* \mid |v| \in S\}$ is regular for any semilinear set $S$. However, due to our concise representation, transforming a simple length constraint into a semilinear length constraints, or transforming a semilinear length constraint into a $\DFA$ representation may cause an exponential size increase. 

By $\matchProb_{\reg}$ and $\matchProb_{\semi}$ we denote the problem $\matchProb$, where all gap constraints are regular constraints or semilinear length constraints, respectively.

For semilinear length constraints, we state the following helpful algorithmic observation.

\begin{lemma}\label{semilinearPreprocessLemma}
For a semilinear set $S$ and an $n \in \mathbb{N}$, we can compute in time $\bigo(n\size(S))$ a data structure that, for every $x \in [n]$, allows us to answer whether $x \in S$ in constant time. 
\end{lemma}

A similar result can be stated for regular constraints.
\begin{lemma}\label{regularPreprocessLemma}
For a regular language $L\subseteq \Sigma^*$, given by a DFA $A$, accepting $L$, and a word $w\in \Sigma^*$, of length $n$, we can compute in time $\bigo(n^2 \log \log n + \size(A))$ a data structure that, for every $i,j\in [n]$, allows us to answer whether $w[i..j] \in L$ in constant time. 
\end{lemma}

\begin{remark}\label{sizeSemiLin}
For every instance $(p, w, \mathcal{C})$ of $\matchProb_{\semi}$, we assume that $\size(C) \leq |w|$ for every $C \in \mathcal{C}$. This is justified, since without changing the solvability of the $\matchProb_{\semi}$ instance, any semilinear constraint defined by some semilinear set $S$ can be replaced by a semilinear constraint defined by the semilinear set $S \cap \{0, 1, 2, \ldots, |w|\}$, which is represented by at most $|w|$ integers.\par
Moreover, we assume $\size(C) \leq |w|^2$ for every regular constraint $C$ for similar reasons. More precisely, a regular constraint defined by a DFA $M$ can be replaced by a constraint defined by a DFA $M'$ that accepts $\{w' \in L(M) \mid w' \text{ is a factor of $w$}\}$. The number of states and, in fact, the size of such a DFA $M'$ can be upper bounded by $|w|^2$. For instance, the DFA $M'$ can be the trie of all suffixes of $w$ (constructed as in \cite{DoriL06}), with the final states used to indicate which factors of $w$ are valid w.r.t. $C$.\par 
We emphasise that working under these assumptions allows us to focus on the actual computation done to match constrained subsequences rather than on how to deal with over-sized constraints.
\end{remark}

\medskip

\section{Complexity of the Matching Problem: Initial Results}\label{sec:complexity}

As parameters of the problem $\matchProb$, we consider the length $|p|$ of the pattern $p$ to be embedded as a subsequence, the number $|\mathcal{C}|$ of gap constraints, the gap description size $\gapsize{\mathcal{C}} = \max\{\size(C) \mid C \in \mathcal{C}\}$, and the alphabet size $|\Sigma|$. Recall that, by assumption, $\mathcal{C}$ is always non-empty, which means that neither $|\mathcal{C}|$ nor $\gapsize{\mathcal{C}}$ can be zero.

For any $\matchProb$-instance, we have that $|\mathcal{C}| \leq |p|^2$. Consequently, if $|p|$ is constant or considered a parameter, so is $|\mathcal{C}|$. This means that an upper bound with respect to parameter $|\mathcal{C}|$ also covers the upper bound with respect to parameter $|p|$, and a lower bound with respect to parameter $|p|$ also covers the lower bound with respect to parameter $|\mathcal{C}|$. Consequently, it is always enough to just consider at most one of these parameters.

For all possible parameter combinations, we can answer the respective complexity for $\matchProb_{\reg}$ and $\matchProb_{\semi}$ (both when the considered parameters are bounded by a constant, or treated as parameters in terms of parameterised complexity). 

From straightforward brute-force algorithms, we can conclude the following upper bounds (see the Appendix for details).

\begin{theorem}\label{UpperBoundsTheorem}
$\matchProb_{\reg}$ and $\matchProb_{\semi}$ can be solved in polynomial time for constant $|\mathcal{C}|$. Moreover, $\matchProb_{\reg}$ is fixed parameter tractable for the combined parameter $(|p|, \gapsize{\mathcal{C}})$.\looseness=-1
\end{theorem}

These upper bounds raise the question whether $\matchProb_{\reg}$ and $\matchProb_{\semi}$ are fixed parameter tractable for the single parameter $|p|$, or whether $\matchProb_{\semi}$ is at least also fixed parameter tractable for the combined parameter $(|p|, \gapsize{\mathcal{C}})$, as in the case of $\matchProb_{\reg}$. Both these questions can be answered in the negative by a reduction from the $\cliqueProb$ problem.\looseness=-1

For the $k$-$\cliqueProb$ problem, we get an undirected graph $G = (V, E)$ and a number $k \in [|V|]$, and we want to decide whether there is a clique of size at least $k$, i.\,e., a set $K \subseteq V$ with $|K| \geq k$ and, for every $u, v \in K$ with $u \neq v$, we have that $\{u, v\} \in E$. It is a well-known fact that $k$-$\cliqueProb$ is $\wclass[1]$-hard. We will now sketch a parameterised reduction from $k$-$\cliqueProb$ to $|p|$-$\matchProb_{\semi}$ and to $|p|$-$\matchProb_{\reg}$ (see the appendix for details).

Let $G = (V, E)$ with $|V| = n$ be a graph represented by its adjacency matrix $A = (a_{i,j})_{1 \leq i, j \leq n}$ (we assume that $a_{i,i} = 1$ for every $i \in [n]$), and let $k \in [|V|]$. It can be easily seen that $G$ has a $k$-clique, if the $k \times k$ matrix containing only 1's is a principal submatrix of $A$, i.e., a submatrix where the set of remaining rows is the same as the set of remaining columns. This can be described as embedding $ p =01^{k^2} 0$ as a subsequence into $w = 0 a_{1,1} a_{1,2} \cdots a_{1,n} a_{2,1} \cdots a_{2,n} \cdots a_{n,1} \cdots a_{n,n} 0$. However, the corresponding embedding $e$ must be such that the complete $i^{th}$ $(1^{k})$-block of $p$ is embedded in the same $a_{s_i, 1} a_{s_i, 2} \cdots a_{s_i, n}$ block of $w$, for some $s_i$. Furthermore, for every $i \in [k]$, the first $1$ of the $i^{th}$ $(1^{k})$-block must be mapped to the $(s_1)^{\text{th}}$ $1$ of $a_{s_i, 1} a_{s_i, 2} \cdots a_{s_i, n}$, the second $1$ of the $i^{th}$ $(1^{k})$-block must be mapped on the $(s_2)^{\text{th}}$ $1$ of $a_{s_i, 1} a_{s_i, 2} \cdots a_{s_i, n}$, and so on. In other words, $1 \leq s_1 < s_2 < \ldots < s_k \leq n$ are the rows and columns where we map the `all-$1$'-submatrix; thus, $\{v_{s_1}, v_{s_2}, \ldots, v_{s_k}\}$ is the clique. Obviously, we have to use the semilinear length constraints to achieve this synchronicity. \looseness=-1

We first force $e(1) = 1$ and $e(k^2 + 2) = n^2 + 2$ by constraint $(1, k^2 + 2, L(n^2))$. In the following, we use $(i, j)_k$ and $(i, j)_n$ to refer to the position of the entries in the $i$-th row and $j$-th column of the flattened matrices in $p$ or $w$ respectively (e.\,g. $w[(i, j)_n] = a_{ij}$). In order to force that $e((i, i)_k) = (s_i, s_i)_n$ for every $i \in [k]$ and some $s_i \in [n]$, we use constraints $(1, (i, i)_k, L(0; n + 1))$, $i \in [k]$ (i.\,e., the first $0$ of $p$ is mapped to the first $0$ of $w$, and then we allow only multiples of $n + 1$ between the images of $(i, i)_k$ and $(i+1, i+1)_k$). Next, we establish the synchronicity between the columns by requiring that the gap between $e((i, j)_k)$ and $e((i+1, j)_k)$ has a size that is one smaller than a multiple of $n$, which is done by constraints $((i, j)_k, (i+1, j)_k, L(n-1; n))$, $i \in [k-1]$, $j \in [k]$. Finally, the constraints $((i, 1)_k, (i, k)_k, \{0\} \cup [n-1])$, $i \in [k]$, force all $e((i, 1)_k), e((i, 2)_k), \ldots, e((i, k)_k)$ into the same block $a_{s_i, 1} a_{s_i, 2} \cdots a_{s_i, n}$. Note that in order to show the last step, we also have to argue with the previously defined constraints for synchronising the columns (more precisely, we have to show that the $e((i, 1)_k), \ldots, e((i, k)_k)$ cannot overlap from one row in the next one).

This is a valid reduction from $k$-$\cliqueProb$ to $|p|$-$\matchProb_{\semi}$ (note that $|p| = k^2 + 2$). We can strengthen the reduction in such a way that all constraints have even constant size (note that the constraints $((i, 1)_k, (i, k)_k, \{0\} \cup [n-1])$ are the only non-constant sized constraints, since $\{0\} \cup [n-1]$ is a semilinear set of size $n$). To this end, we observe that for fixed $s_1$ and $s_k$, all gap sizes $e((i, k)_k) - e((i, 1)_k)$ are the same, independent from $i$, namely $s_k - s_1$. Thus, we can turn the reduction into a Turing reduction by guessing this value $d = s_k - s_1$ and then replace each \emph{non-constant} $((i, 1)_k, (i, k)_k, \{0\} \cup [n-1])$ by $((i, 1)_k, (i, k)_k, L(d))$. 

In order to obtain a reduction to $\matchProb_{\reg}$, we can simply represent all semilinear constraints as regular constraints. Obviously, the corresponding DFAs are not of constant size anymore. In summary, this yields the following result.

\begin{theorem}\label{cliqueTheorem}
$\matchProb_{\semi}$ parameterised by $|p|$ is $\wclass[1]$-hard, even for constant $\gapsize{\mathcal{C}}$ and binary alphabet $\Sigma$, and $\matchProb_{\reg}$ parameterised by $|p|$ is $\wclass[1]$-hard, even for binary alphabet $\Sigma$.
\end{theorem}

The lower bound for $\matchProb_{\reg}$ is weaker, since the parameter $\gapsize{\mathcal{C}}$ is not constant in the reduction. At least for a reduction from $k$-$\cliqueProb$, this is to be expected, due to the fact that $\matchProb_{\reg}$ is fixed parameter tractable with respect to the combined parameter $(|p|, \gapsize{\mathcal{C}})$. So this leaves one relevant question open: Can $\matchProb_{\reg}$ be solved in polynomial time, if the parameter $\gapsize{\mathcal{C}}$ is bounded by a constant? We can answer this in the negative by a reduction from a variant of the SAT-problem, which we shall sketch next (see the Appendix for full details).

For the problem $1$-in-$3$-3SAT, we get a set $A = \{x_1, x_2, \ldots, x_n\}$ of variables and clauses $c_1, c_2, \ldots, c_m \subseteq A$ with $|c_i| = 3$ for every $1 \leq i \leq m$. The task is to decide whether there is a subset $B \subseteq A$ such that $|c_i \cap B| = 1$ for every $1 \leq i \leq m$. For the sake of concreteness, we also set $c_j = \{x_{\ell_j, 1}, x_{\ell_j, 2}, x_{\ell_j, 3}\}$ for every $j \in [m]$, and with $x_{\ell_j, 1} < x_{\ell_j, 2} < x_{\ell_j, 3}$ for some order `$<$' on $A$. \par
We transform such an $1$-in-$3$-3SAT-instance into two strings $u_A = (\tb \#)^n (\tb \#)^m$ and $v_A = (\tb^2 \#)^n (\tb^3 \#)^m$. For every $i \in [n]$, the $i^{\text{th}}$ $\tb$-factor of $u_A$ and the $i^{\text{th}}$ $\tb^2$-factor of $v_A$ are called \emph{$x_i$-blocks}. Analogously, we denote the last $m$ $\tb$-factors of $u_A$ and the last $m$ $\tb^3$-factors of $v_A$ as \emph{$c_j$-blocks} for $j \in [m]$. 

Obviously, if $u_A \subseq_e v_A$ for some embedding $e : [|u_A|] \to [|v_A|]$, then the single $\tb$ of $u_A$'s $x_i$-block is mapped to either the first or the second $\tb$ of $v_A$'s $x_i$-block, and the single $\tb$ of $u_A$'s $c_j$-block is mapped to either the first or the second or the third $\tb$ of $v_A$'s $c_j$-block.
Thus, the embedding $e$ can be interpreted as selecting a set $B \subseteq A$ (where mapping $u_A$'s $x_i$-block to the \emph{second} $\tb$ of $v_A$'s $x_i$-block is interpreted as $x_i \in B$), and selecting either the first or second or third element of  $c_j$ (depending on whether $u_A$'s $c_j$-block is mapped to the first, second or third $\tb$ of $v_A$'s $c_j$-block).
We can now introduce a set of regular gap constraints that enforce the necessary synchronicity between $B$ and the elements picked from the clauses: Assume that $x_i$ is the $p^{\text{th}}$ element of $c_j$. If $e$ maps $u_A$'s $x_i$-block to the second $\tb$ of $v_A$'s $x_i$-block, then $e$ must map $u_A$'s $c_j$-block to the $p^{\text{th}}$ $\tb$ of $v_A$'s $c_j$-block, and if $e$ maps $u_A$'s $x_i$-block to the first $\tb$ of $v_A$'s $x_i$-block, then $e$ must map $u_A$'s $c_j$-block to the $q^{\text{th}}$ $\tb$ of $v_A$'s $c_j$-block for some $q \in \{1, 2, 3\} \setminus \{p\}$. For example, if $x_{\ell_j, 2} = x_i$ for some $i \in [n]$ and $j \in [m]$, i.\,e., the second element of $c_j$ is $x_i$, then we add a regular gap constraint $(i', j', L_{i', j'})$, where $i'$ and $j'$ are the positions of $u_A$'s $x_i$-block and $u_A$'s $c_j$-block, and $L_{i', j'} = \{\tb w \#, \# w \# \tb, \tb w \# \tb \tb \mid w \in \{\tb, \#\}^*\}$. If $e$ maps $u_A$'s $x_i$-block to the second $\tb$ of $v_A$'s $x_i$-block, then the gap between positions $i'$ and $j'$ must start with $\#$; thus, due to the gap constraint, it must be of the form $\# w \# \tb$, which is only possible if $e$ maps $u_A$'s $c_j$-block to the second $\tb$ of $v_A$'s $c_i$-block. If $e$ maps $u_A$'s $x_i$-block to the first $\tb$ of $v_A$'s $x_i$-block, then the gap must be of the form $\tb w \#$ or $\tb w \# \tb \tb$, which means that $e$ must map $u_A$'s $c_j$-block to the first or third $\tb$ of $v_A$'s $c_i$-block. Similar gap constraints can be defined for the case that $x_i$ is the first or third element of $c_j$. Hence, we can define gap constraints such that there is an embedding $e : [|u_A|] \to [|v_A|]$ with $u_A \subseq_e v_A$ and $e$ satisfies all the gap constraints if and only if the $1$-in-$3$-3SAT-instance is positive. Independent of the actual $1$-in-$3$-3SAT-instance, the gap languages can be represented by DFAs with at most $8$ states. This shows the following result.

\begin{theorem}\label{SATTheorem}
$\matchProb_{\reg}$ is $\npclass$-complete, even for binary alphabet $\Sigma$ and with gap-constraints that can be represented by $\DFA$s with at most $8$ states.
\end{theorem}

\section{Complexity of the Matching Problem: A Finer Analysis}\label{sec:vsn-and-intervals}

\subparagraph*{Two representations of constraints.} We start by defining two (strongly related) natural representations of the set of constraints which is given as input to the matching problem. Intuitively, both these representations facilitate the understanding of a set of constraints. Then, we see how these representations can be used to approach the $\matchProb$ problem.

\subparagraph*{The Interval Structure of Sets of Constraints.} For a constraint $C = (a, b, L)$ we define $\interval(C) \colonequals [a, b-1]$. If we have another constraint $C' = (a', b', L')$ (with $a \neq a' $ or $b \neq b'$), we say that $C$ is contained in $C'$, written $C \sqsubset C'$, if $\interval(C) \subsetneq \interval(C')$. Because this order is derived from the inclusion order $\subsetneq$, it is also a strict partial order. We denote the corresponding covering relation with $\sqsubsetdot$. Furthermore, we say that $C$ and $C'$ intersect, if $\interval(C) \cap \interval(C') \neq \emptyset$, and they are not comparable w.r.t.\ $\sqsubset$, i.e., neither of them contains the other. Importantly, because we have $b \notin \interval(C)$, in the case $b = a'$, the constraints $C$ and $C'$ do not intersect. 

\subparagraph*{The Graph Structure of Sets of Constraints.} For a string $p$, of length $m$, and a set of constraints ${\mathcal C}$ on $p$, we define a graph $G_p=(V_p,E_p)$ as follows. The set of vertices $V_p$ of $G_p$ is the set of numbers $\{1,\ldots,m\}$, corresponding to the positions of $p$. We define the set of undirected edges $E_p$ as follows: $E_p=\{(a,b)\mid (a,b,L)\in {\mathcal C}\}\cup \{(i,i+1)\mid i\in [m-1]\}\cup\{(1,m)\}$. Note that in the case when we have a constraint $C=(i,i+1, L)$, for some $i\in [m-1]$ and $L$ (respectively, a constraint $C(1,m,L)$) we will still have a single edge connecting the nodes $i$ and $i+1$ (respectively, $1$ and $m$) as $E_p$ is constructed by a set union of two sets, so the elements in the intersection of the two sets are kept only once in the result. Moreover, we can define a labelling function on the edges of $G$ by the following three rules: $\lab((i,j))=C$, if there exists $C\in {\mathcal C}$ with $C=(i,j,L)$; $\lab((i,i+1))=(i,i+1,\Sigma^*)$, for $i\in [m-1]$, if there exists no $C\in {\mathcal C}$ such that $C=(i,i+1,L)$; and $\lab((1,m))=(1, m, \{w\in \Sigma^*||w|\geq m-2\})$, if there exists no $C\in {\mathcal C}$ such that $C=(1,m,L)$. Clearly, all respective labels can be expressed trivially both as regular languages or as semilinear sets. Intuitively, the edges of $G$ which correspond to constraints of ${\mathcal C}$ are labelled with the respective constraints. The other edges have trivial labels: the label of the edges of the form $(i,i+1)$ express that, in an embedding of $p$ in the string $v$ (as required in $\matchProb$), the embedding of position $i+1$ is to the right of the embedding of position $i$; the label of edge $(1,m)$ simply states that at least $m-2$ symbols should occur between the embedding of position $1$ and the embedding of position $m$, therefore allowing for the entire pattern to be embedded. 

Further, it is easy to note that this graph admits a Hamiltonian cycle, which traverses the vertices $1,2,\ldots,m$ in this order. In the following, we also define two-dimensional drawing of the graph $G_p$ as a {\em half-plane arc diagram}: the vertices $1,2,\ldots, m$ are represented as points on a horizontal line $\ell$, with the edges $(i,i+1)$, for $i\leq m-1$ being segments of unit-length on that line, spanning between the respective vertices $i$ and $i+1$; all other edges $(i,j)$ are drawn as semi-circles, whose diameter is equal to the length of $(j-i)$, drawn in the upper half-plane with respect to the line $\ell$. In the following, we will simply call this diagram arc-diagram, without explicitly recalling that all semicircles are drawn in the same half-plane with respect to $\ell$. In this diagram associated to $G_p$, we say that two edges cross if and only if they intersect in a point of the plane which is not a vertex of $G_p$. See also \cite{DjidjevV03} and the references therein for a discussion on this representation of graphs. 

In \autoref{fig_constraint_relations} we see an example: We have $p=\mathtt{xyzyx}$ and $\mathcal{C}=\{C=(1,3,L_1), C'=(1,4,L_1), C''=(3,5,L_2)\}$. Thus, $|p|=5$ and $G_p$ has the vertices $\{1,\ldots,5\}$ and the edges $(1,3),(1,4),(3,5)$ corresponding to the constraints of $\mathcal{C}$, as well as the edges $(1,2),$ $(2,3),$ $(3,4),$ $(4,5),$ $(1,5)$ (note that the edges are undirected, and the trivial labels are omitted for the sake of readability). The figure depicts the arc diagram associated to this graph (with the semi-circles flattened a bit, for space reasons). In the interval representation of these constraints, we can see that $C$ is contained in $C'$ (as $\interval(C)=[1, 2] \subsetneq \interval(C')=[1, 3]$). Furthermore, $C'$ and $C''$ intersect ($[1, 3] \cap [3, 4] \neq \emptyset$), while $C$ and $C''$ do not ($[1, 2] \cap [3, 4] = \emptyset$). Note that two constraints (such as $C$ and $C''$ in our example) might not intersect (according to the interval representation), although the edges that correspond to them in the graph representation share a common vertex; in particular, two constraints intersect if and only if the corresponding edges cross.

\begin{figure}[h]
	\centering
{
	\begin{tikzpicture}
		\node (P) {$p = {}$};
        \node (P1) [right of = P]{\texttt{1}};
        \node (P11) [below of = P1, node distance=10pt]{\texttt{x}};
		\node (P2) [right of = P1]{\texttt{2}};
        \node (P21) [right of = P11]{\texttt{y}};
		\node (P3) [right of = P2]{\texttt{3}};
        \node (P31) [right of = P21]{\texttt{z}};
		\node (P4) [right of = P3]{\texttt{4}};
        \node (P41) [right of = P31]{\texttt{y}};
		\node (P5) [right of = P4]{\texttt{5}};
        \node (P51) [right of = P41]{\texttt{x}};
				\draw (P1) to node[midway, below] {} (P2);
				\draw (P2) to node[midway, below] {} (P3);
				\draw (P3) to node[midway, below] {} (P4);
				\draw (P4) to node[midway, below] {} (P5);
		\draw (P1.north) to[bend left = 45] node[midway, below] {\scriptsize $C$} (P3.north);
		\draw (P1.north) to[bend left = 90] node[midway, below] {} (P5.north);
		\draw (P1.north) to[bend left = 45] node[midway, above] {\scriptsize $C'$} (P4.north);
        \draw (P3.north) to[bend left = 45] node[midway, above] {\scriptsize $C''$} (P5.north);
	\end{tikzpicture}
}
    \caption{Relations between constraints}\label{fig_constraint_relations}
\end{figure}
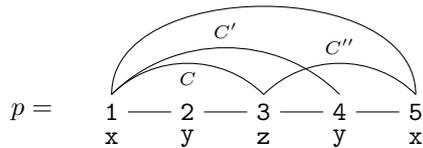

The two representations defined above are clearly very strongly related. However, the graph-representation allows us to define a natural structural parameter for subsequences with constraints, while the interval-representation will allow us to introduce a class of subsequences with constraints which can be matched efficiently.

In the following, by $\matchProb_{{\mathcal L}, {\mathcal G}}$ we denote the problem $\matchProb$, where all gap constraints are from the class of languages ${\mathcal L}$, with ${\mathcal L}\in \{\reg,\semi\}$, and the graphs corresponding to the input gap constraints are all from the class ${\mathcal G}$. 

\subparagraph{Vertex separation number and its relation to $\matchProb$.} Given a linear ordering $\sigma=(v_1,\ldots,v_m)$ of the vertices of a graph $G$ with $m$ vertices, the vertex separation number of $\sigma$ is the smallest number $s$ such that, for each vertex $v_i$ (with $i \in [m]$), at most $s$ vertices of $v_1,\ldots, v_{i-1}$ have some $v_j$, with $j\geq i$, as neighbour. The vertex separation number $\vsn(G)$ of $G$ is the minimum vertex separation number over all linear orderings of $G$. The vertex separation number was introduced in \cite{EllisST83} (see also \cite{EllisST94} and the references therein) and was shown (e.g., in \cite{Bodlaender98}) to be equal to the pathwidth of $G$. 

Let us briefly overview the problem of computing the vertex separation number of graphs. Firstly, we note that checking, given a graph $G$ with $n$ vertices and a number $k$, whether $\vsn(G)\leq k$ is $\npclass$-complete, as it is equivalent to checking whether the pathwidth of $G$ is at most $k$. We can show (see \cref{lem:svnham}) that this problems remains intractable even when we restrict it to the class of graphs with a Hamiltonian cycle, and this cycle is given as input as well. \looseness=-1

However, this problem is linear fixed parameter tractable w.r.t. the parameter $k$: deciding, for a given graph $G$ with $n$ vertices and a constant number $k$, if $\vsn(G)\leq k$ and, if so, computing a linear ordering of the vertices with vertex separation number at most $k$ can be solved in $O(n)$ time, where the constant hidden by the $O$-notation depends superexponentially on $k$. This follows from the results of \cite{Bodlaender98}, where the relation between pathwidth and vertex separation number is established, and \cite{Bodlaender96}, where it is shown that, for constant $k$, one can check in linear time $O(n)$ if a graph with $n$ vertices has pathwidth at most $k$, and, if so, produce a path decomposition of that graph of width at most $k$. 

For a constant $k$, let ${\mathcal V}_k$ be the class of all graphs $G$ with $\vsn(G)\leq k$. We can show the following meta-theorem.

\begin{theorem}\label{thm:constant-vsn}
Let $k\geq 1$ be a constant integer and let ${\mathcal L}\in\{\semi,\reg\}$. Then, $\matchProb_{{\mathcal L}, {\mathcal V}_k}$ can be solved in polynomial time: $O(m^2n^{k+1})$, in the case of $\semi$-constraints, and $O(m^2n^{k+1}+ m^2n^2\log \log n)$, in the case of $\reg$-constraints. Moreover, $\matchProb_{{\mathcal L}, {\mathcal V}_k}$ parameterised by $k$ is $\wclass[1]$-hard.
\end{theorem}
A full proof of this theorem is given in Appendix \ref{sec:appendixVSN}; here we only sketch it. The matching algorithm implements a dynamic programming approach. We choose an ordering of the vertices of the graph representing ${\mathcal C}$, with $\vsn$ at most $k$. These vertices are, in fact, positions of $p$, so we find, for $q$ from $1$ to $m$, embeddings for the first $q$ positions of this ordering in $w$, such that all the constraints involving only these positions are fulfilled. Given that the $\vsn$ of the respective ordering is bounded by $k$, we can compute efficiently the embeddings of the first $q$ positions by extending the embeddings for the first $q-1$ positions, as, when considering a new position of our ordering, and checking where it can be embedded, only $k$ of the previously embedded positions are relevant. As such, the embeddings of the first $q-1$ positions of the ordering, which are relevant for computing the embeddings of its first $q$ positions, can be represented using an $O(n^k)$ size data-structure, and processed in $O(n^k\poly(n,m))$ time. This leads to a polynomial time algorithm, with a precise runtime as stated above. The lower bound follows from the reduction showing Theorem \ref{cliqueTheorem}. \looseness=-1

\subparagraph{Non-intersecting constraints and $\matchProb$.}

We have shown that $\matchProb$ can be solved efficiently if the input gap constraints are represented by graphs with bounded $\vsn$. However, while this condition is sufficient to ensure that $\matchProb$ can be solved efficiently (as long as the constraint-languages are in $\pclass$), it is not necessary. We will exhibit in the following a class ${\mathcal H}$ of gap constraints, which contains graphs with arbitrarily large $\vsn$, and for which $\matchProb_{{\mathcal L},{\mathcal H}}$ can be solved in polynomial time, for ${\mathcal L}\in \{\reg,\semi\}$.

More precisely, in the following, we will consider the class ${\mathcal H}$ of non-intersecting gap constraints. That is, we consider $\matchProb$ where the input consists of two strings $v \in \Sigma^*$ and $p \in \Sigma^*$ and a non-empty set $\mathcal{C}$ of gap constraints, where for any $C,C'\in \mathcal{C}$ we have that $\interval(C)\cap \interval(C') \in \{\interval(C), \interval(C'), \emptyset\}$. It is immediate that ${\mathcal H}$, the class of non-intersecting gap constraints, can be described as the class of gap constraints which are represented by outerplanar graphs which have a Hamiltonian cycle: the arc diagram constructed for these gap constraints is already outerplanar: if $C=(a,b,L)$ and $C'=(a',b',L')$ are two non-intersecting constraints of some set of constraints ${\mathcal C}$, then, in the graph representation of this set of constraints based on arc diagrams, the edges $(a,b)$ and $(a',b')$ do not cross (although they might share a common vertex). 

Moreover, an outerplanar graph which admits a Hamiltonian cycle can be represented canonically as an arc diagram of a set of non-intersecting gap constraints. It is a folklore result that if an outerplanar graph has a Hamiltonian cycle then the outer face forms its unique Hamiltonian cycle. Moreover, every drawing of a graph in the plane may be deformed into an arc diagram without changing its number of crossings \cite{Nicholson68}, and, in the case of outerplanar graphs this means the following. For an outerplanar graph $G$, we start with a drawing of $G$ witnessing its outerplanarity. Assume that, after a potential renaming, there exists a traversal of the Hamiltonian cycle of the outerplanar graph (i.e., of its outer face) which consists of the vertices $1,2,\ldots,n$, in this order. We simply reposition these vertices, in the same order $1,2,\ldots,n$, on a horizontal line $\ell$, such that consecutive vertices on the line are connected by edges of length $1$, and then the edge $(1,n)$ is deformed so that it becomes a semicircle of diameter $n-1$ connecting the respective vertices, in the upper half-plane w.r.t. $\ell$. Further, each edge $(a,b)$ is deformed to become a semicircle above the line of the vertices, whose diameter equals the distance between vertex $a$ and vertex $b$. By the result of \cite{Nicholson68}, the edges of this graph do not cross in this representation, as the initial graph was outerplanar. But the resulting diagram is, clearly, the arc diagram corresponding to a set of non-intersecting gap constraints.

Based on the above, we can make a series of observations. Firstly, as one can recognize outerplanar graphs in linear time \cite{Wiegers86}, we can also decide in linear time whether a set of constraints is non-crossing. Secondly, according to \cite{CoudertHS07}, there are outerplanar graphs with arbitrarily large pathwidth, so with arbitrarily large $\vsn$. This means that there are sets of non-intersecting gap constraints whose corresponding graph representations have arbitrarily large $\vsn$. Thirdly, the number of constraints in a set of non-intersecting gap constraints is linear in the length of the string constrained by that set (as the number of edges in an outerplanar graph with $n$ vertices is at most $2n-3$). 

We can show the following theorem (here we just sketch the proof, and only in the case of $\semi$-constraints, as the $\reg$-constraints case is identical; for a full version see Appendix \ref{sec:NonIntersectUB}).
\begin{theorem}\label{thm:NonIntersectUB}
$\matchProb_{\semi, {\mathcal H}}$ can be solved in time $O(n^\omega K)$ and $\matchProb_{\reg, {\mathcal H}}$ can be solved in time $O(n^\omega K + n^2 K \log\log n)$, where $K$ is the number of constraints in the input set of constraints ${\mathcal C}$ and $O(n^\omega)$ is the time needed to multiply two boolean matrices of size $n\times n$.
\end{theorem}
As said, we sketch the algorithm solving $\matchProb_{{\semi}, {\mathcal H}}$ in the stated complexity. Assume that the input words are $w\in \Sigma ^n$ and $p\in \Sigma^m$, and ${\mathcal C}=(C_1,\ldots,C_K)$, with $C_i=(a_i,b_i,L_i)$ for $i\in [K]$. Firstly, we add a constraint $C=(1,m,L(m-2,1))$ to ${\mathcal C}$ if it does not contain any constraint having the first two components $(1,m)$. So, in the following, we will assume w.l.o.g. that such a constraint $(1,m,\cdot)$ always exists in ${\mathcal C}$. Moreover, the number of constraints in ${\mathcal C}$ is $O(m)$ (as the graph representation of ${\mathcal C}$ is outerplanar). 

As a first phase in our algorithm, we build the data structures from Lemma \ref{semilinearPreprocessLemma}. Hence, by Remark \ref{sizeSemiLin}, after an $O(n^2K)$-time preprocessing we can answer queries ``is $w[i..j]\in L_k$?'' in $O(1)$ time, for all $i,j\in [n], k\in [K]$.

After this, the algorithm proceeds as follows. Because the set of constraints ${\mathcal C}$ is non-intersecting, one can build in linear time the Hasse-diagram of the set of intervals associated with the set of constraints ${\mathcal C}$ (w.r.t. the interval-inclusion relation), and this diagram is a tree, whose root corresponds to the single constraint of the form $(1,m,\cdot)$. Further, the algorithm uses a dynamic programming strategy to find matches for the constraints of ${\mathcal C}$ in a bottom-up fashion with respect to the Hasse-diagram of this set. The algorithm maintains the matches for each constraint $C=(a,b,L)$ as a Boolean $n\times n$ matrix, where the element on position $(i,j)$ of that matrix is true if and only if there exists a way to embed $p[a..b]$ in $w[i..j]$, such that $p[a]$ is mapped to $w[i]$ and $p[b]$ to $w[j]$ in the respective embedding, and this embedding also fulfils $C$ and all the constraints occurring in the sub-tree of root $C$ in the Hasse-diagram. This matrix can be computed efficiently, by multiplying the matrices corresponding to the children of $C$ (and a series of matrices corresponding to the unconstrained parts of $p[a..b]$). As the number of nodes in this tree is $O(K)$,  the whole process of computing the respective matrices for all nodes of the tree requires $O(K)$ matrix multiplications, i.e., $O(n^\omega K)$ time in total. Finally, one needs to see if there is a match of $p[1..m]$ to some factor of $v$, which can be checked in $O(n^2)$ by simply searching in the matrix computed for the root of the diagram. 

The following lower bound is shown by a reduction from $3$-$\OV$ (see Appendix \ref{sec:NonIntersectLB}). \looseness=-1
\begin{theorem}\label{thm:NonIntersectLB}
For ${\mathcal L}\in\{\reg,\semi\}$, then $\matchProb_{{\mathcal L}, {\mathcal H}}$, cannot be solved in $\mathcal{O}(|w|^g|\mathcal{C}|^h)$ time with $g + h < 3$, unless \textsf{SETH} fails.
\end{theorem}
The reduction proving this hardness result works as follows. In $3$-$\OV$, we are given three sets $A = \{\vec{a}_1, \dots, \vec{a}_n\}$, $B = \{\vec{b}_1, \dots, \vec{b}_n\}$ and $C = \{\vec{c}_1, \dots, \vec{c}_n\}$ with elements from $\{0, 1\}^d$, and want to determine whether there exist $i^*, j^*, k^* \in [n]$, such that $\sum^j_{\ell=1}\vec{a}_{i^*}[\ell] \cdot \vec{b}_{j^*}[\ell] \cdot \vec{c}_{k^*}[\ell]=0$. This is achieved by encoding our input sets over a constant size alphabet, via two functions $\mathsf{C}_p$ and $\mathsf{C}_w$, into a pattern $p$ and a text $w$, respectively, as well as a set of constraints $\mathcal{C}$, such that the answer to the $3$-$\OV$ problem is positive if and only if $p$ is a $\mathcal{C}$-subsequence of $w$. Basically, the encoding of each $d$-dimensional vector of $A$ (respectively, $B$ and $C$) is done via $\mathsf{C}_p$ (respectively, $\mathsf{C}_w$), in such a way that (when no constraints are considered) $\mathsf{C}_p(\vec{v})$ is a subsequence of $\mathsf{C}_w(\vec{v}')$ for any $\vec{v}, \vec{v}' \in \{0, 1\}^d$. Further, $\overline{\mathsf{C}}_p$ and $\overline{\mathsf{C}}_w$ are mirrored versions of these encodings (both for bits and vectors), where the order of the characters in the output is inverted. We can then use the original encoding for one part of the pattern and the text and the mirrored encoding for the other part. Then, we encode the set $A$ in $p \colonequals \overline{\mathsf{C}}_p(\vec{a}_n)\, \dots\, \overline{\mathsf{C}}_p(\vec{a}_1)\, \S\, \mathsf{C}_p(\vec{a}_1)\, \dots\, \mathsf{C}_p(\vec{a}_n)$ and the sets $B$ and $C$ in $w \colonequals \overline{w}_0\, \#\, \overline{\mathsf{C}}_w(\vec{c}_n)\, \dots\, \overline{\mathsf{C}}_w(\vec{c}_1)\, \#\, \overline{w}_0\, \S\, w_0\, \#\, \mathsf{C}_w(\vec{b}_1)\, \dots\, \mathsf{C}_w(\vec{b}_n)\, \#\, w_0$ (where $\overline{x}$ denotes the mirror image of $x$, and $w_0$ is a suitably choosen padding). To finalise our construction, we can define constraints which ensure that an embedding of $p$ in $w$ is possible if and only if there exist some  $i^*, j^*, k^*$ such that $\mathsf{C}_p(a_{i^*})$ is embedded in $\mathsf{C}_w(b_{j^*})$, and $\overline{\mathsf{C}}_p(a_{i^*})$ in $\overline{\mathsf{C}}_p(c_{k^*})$, while all the other strings $\overline{\mathsf{C}}_p(a_{t})$ are embedded in the paddings. Moreover, additional constraints ensure that the simultaneous embedding of $\mathsf{C}_p(a_{i^*})$ in $\mathsf{C}_w(b_{j^*})$ and of $\overline{\mathsf{C}}_p(a_{i^*})$ in $\overline{\mathsf{C}}_p(c_{k^*})$ is only possible if and only if for each component $u\in [d]$ with $a_{i^*}[u]\neq 0$, we have that $b_{j^*}[u]=0$ or $c_{k^*}[u]=0$. \looseness=-1

We conclude by noting that, while this is not a tight lower bound with respect to the upper bound shown in \cref{thm:NonIntersectUB}, finding a polynomially stronger lower bound (i.e., replacing in the statement of Theorem \ref{thm:NonIntersectLB} the condition $g + h < 3$ with $g + h < 3+\delta$, for some $\delta>0$) would show that matrix multiplication in quadratic time is not possible, which in turn would solve a well-researched open problem. Indeed, the algorithm from \cref{thm:NonIntersectUB} consists in a reduction from $\matchProb_{{\mathcal L}, {\mathcal H}}$ to $O(|C|)$ instances of matrix multiplication, for quadratic matrices of size $|w|$, so a better lower bound would mean that at least one of these multiplications must take more than quadratic time. 


\appendix

\newpage
\section{Contents}

This appendix contains the missing full proofs from the paper, as well as several definitions and basic concepts. Admittedly, some of the proofs of this Appendix could have been presented in a more compact manner. Our feeling was that such a denser presentation might hinder the understanding of these proofs, so we preferred to give here full details (including figures and pseudo-codes). However, we believe that the main ideas behind our proofs are already presented in the main part of the paper, and this (long) Appendix is mostly needed to address the finer, but more tedious, technicalities. 

\section{Additional Definitions}

\medskip

\noindent\textbf{Strict Partial Orders}: Given a set $X$, we call $\mathcal{R} \subseteq X \times X$ a \emph{(homogeneous) relation on $X$}. Then, for any pair $(x, y) \in \mathcal{R}$, we say that $x$ \emph{is related to} $y$. Often, $\mathcal{R}$ is replaced with an operator, for example ``$<$'', and we write $x < y$ if $x$ is related to $y$ and $x \not< y$ otherwise. Two elements $x, y \in X$ are called \emph{comparable}, if $x < y$ or $y < x$. Now $<$ is a \emph{strict partial order}, if it fulfills the following three properties:
\begin{itemize}
	\item \emph{Irreflexivity:} $x \not< x$ for all $x \in X$
	\item \emph{Asymmetry:} $x < y \implies y \not< x$ for all $x, y \in X$
	\item \emph{Transitivity:} $x < y \land y < z \implies x < z$ for all $x, y, z \in X$
\end{itemize}
We then call $(X, <)$ a \emph{partially ordered set}, or \emph{poset} for short. The word ``partial'' indicates that not all pairs of elements $x, y \in X$ must be comparable.

For a strict partial order $<$ on $X$ and two elements $x, y \in X$, we say that y \emph{covers} x, if $x < y$ and there is no $z \in X$, such that $x < z < y$. In this case, we write $x \lessdot y$ and call $\lessdot$ the \emph{covering relation} of $(X, <)$. Now, if we assume that $X$ is finite, the covering relation can be used for a visual representation of the poset:

Note that for a directed graph $G$ with vertices $V$ and edges $E$, we can interpret $E \subseteq V \times V$ as a relation on $V$. Conversely, we can view $(X, \lessdot)$ as a directed graph, called the \emph{Hasse diagram} of the poset $(X, <)$. It can be seen as a intuitive visualization of $(X, <)$ due to the following property:

\begin{lemma}\label{lem_hasse}
	For any $x, y \in X$ with $x \neq y$ it is $x < y$ if and only if there is a path from $x$ to $y$ in the Hasse Diagram of $(X, <)$.
\end{lemma}
\begin{proof}
	If we have the path $x \to z_1 \to \dots \to z_n \to y$, then $x \lessdot z_1 \lessdot \dots \lessdot z_n \lessdot y$, by definition we have $x < z_1 < \dots < z_n < y$ and, by transitivity, $x < y$. Also, if $x < y$ then either $x \lessdot y$ or there exists a $z \in X$ with $x < z < y$ and we can recursively apply the same argument to $x < z$ and $z < y$. Since $X$ is finite and any element can only occur once (because of irreflexivity and transitivity), we will find a sequence $x \lessdot z_1 \lessdot \dots \lessdot z_n \lessdot y$ eventually.
\end{proof}

Because of this lemma and the irreflexivity of $<$ the graph must be acyclic and can thus be drawn so that all edges are pointing upwards. Then, the edges are usually drawn as plain line segments, because their direction is given from their orientation.

This paper addresses, in particular, two strict partial orders on $\mathcal{I}_k \colonequals \{[m, n] \mid m, n \in [k], m \leq n\}$: The \emph{inclusion order} $\subsetneq$, and the \emph{interval order} $<$, which is defined as
\[
	[m_1, n_1] < [m_2, n_2] \quad\logeq\quad n_1 < m_2
\]
As an example, \autoref{fig_hasse_example} shows Hasse diagrams for both the inclusion order (left) and the interval order (right) on $\mathcal{I}_3$.

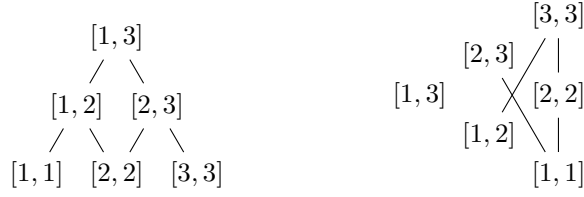
\begin{figure}[h]
	\centering
	\begin{tikzpicture}[node distance=3em]
		\node (13) {$[1, 3]$};
		\node (12) at ([shift=({240:3em})]13) {$[1, 2]$};
		\node (23) at ([shift=({300:3em})]13) {$[2, 3]$};
		\node (11) at ([shift=({240:3em})]12) {$[1, 1]$};
		\node (22) at ([shift=({300:3em})]12) {$[2, 2]$};
		\node (33) at ([shift=({300:3em})]23) {$[3, 3]$};

		\draw (13) -- (12);
		\draw (13) -- (23);
		\draw (12) -- (11);
		\draw (12) -- (22);
		\draw (23) -- (22);
		\draw (23) -- (33);
	\end{tikzpicture}
	\hspace{5em}
	\begin{tikzpicture}[node distance=3em]
		\node (13) {$[1, 3]$};
		\node (12) at ([shift=({330:3em})]13) {$[1, 2]$};
		\node (23) at ([shift=({30:3em})]13) {$[2, 3]$};
		\node (11) at ([shift=({330:3em})]12) {$[1, 1]$};
		\node (22) at ([shift=({30:3em})]12) {$[2, 2]$};
		\node (33) at ([shift=({30:3em})]23) {$[3, 3]$};

		\draw (11) -- (22);
		\draw (11) -- (23);
		\draw (12) -- (33);
		\draw (22) -- (33);
	\end{tikzpicture}
	\caption{Hasse diagrams for inclusion and interval order on $\mathcal{I}_3$}\label{fig_hasse_example}
\end{figure}

\section{Proof Details of Section~\ref{sec:gapConstraints}}

\subsection{Proof of Lemma~\ref{semilinearPreprocessLemma}}

\begin{proof}
Let $L = L(x_{0}; x_{1}, \ldots, x_{l})$ be some linear set. Then, for every $x \in [n]$, we have $x \in L$, iff $x = x_{0}$ or there exists $x' < x$, such that $x' \in L$ and $x - x' \in \{x_{1}, \dots, x_{l}\}$. Hence, we can construct in time $\bigo(n l)$ an $n$-dimensional Boolean vector $\tilde{L}$, where, for every $x \in [n]$, $\tilde{L}[x]$ is true if and only if $x \in L$, by the following algorithm.

\begin{algorithm}[ht]
	\caption{Compute $\tilde{L}$}
	\label{algo_preprocess}
	\begin{algorithmic}
		\Function{preprocess}{$L$}
		\State $\tilde{L} \gets [\texttt{false}, \dots, \texttt{false}]$
		\Comment{$n$ elements}
		\State $\tilde{L}[x_0] \gets \texttt{true}$
		\For{$i \gets x_0, \dots, n$}
		\If{$\tilde{L}[i]$}
		\For{$j \gets 1, \dots, l$}
		\If{$i + x_j \leq n$}
		\State $\tilde{L}[i + x_j] \gets \texttt{true}$
		\EndIf
		\EndFor
		\EndIf
		\EndFor
		\State \Return $\tilde{L}$
		\EndFunction
	\end{algorithmic}
\end{algorithm}

Now let $S = L_1 \cup L_2 \cup \ldots \cup L_k$. By the above algorithm, for all $i\in [k]$, we can compute the vector $\tilde{L}_{i}$ in time $\bigo(n l_i)$.  So, all vectors $\tilde{L}_{i}$, with $i\in [k]$ can be computed in $\bigo(n \sum^k_{i = 1} l_i)$. Then, in time $\bigo(n k)$ we can combine them (using $\lor$) to determine a vector $\tilde{S}$ such that, for every $x \in [n]$, $\tilde{S}[x]$ is true if and only if $x \in S$. In total, this requires time $\bigo(n \sum^k_{i = 1} l_i) = \bigo(n |S|)$. 
\end{proof}

\subsection{Proof of Lemma~\ref{regularPreprocessLemma}}

\begin{proof} We assume that $A$ is given as a list of states, and for each state $q$ of the DFA we are given the list (of size at most $|\Sigma|$) of transitions $(q,a,q')$ (where $q'$ is the target state of the transition from state $q$ with letter $a$), sorted w.r.t. the letter $a$ labelling the respective transition (if this list is not given sorted, we can actually sort it in linear time using radix sort). As $A$ is potentially incomplete, the list of transitions of some state may have less than $|\Sigma|$ transitions. This representation uses $O(\size(A))$ memory. However, with this representation, computing the state reached by $A$ after reading a word $v$ of length $\ell$ takes $O(\ell \log|\Sigma|)$ (as after reading each letter $v[i]$ of $v$, we find the transition that needs to be executed by binary searching it (i.e., searching for the transition labelled with $v[i]$) in the sorted list of transitions leaving the current state).  

We can improve our implementation of DFAs by storing, for each state $q$, the list of all transitions having the form $(q,a,q')$ (where $q'$ is the target state of the transition from state $q$ with letter $a$) as a $y$-fast tries \cite{Willard83}, where the main key of the elements stored in the trie are the letters labelling those transitions, which are drawn from the integer alphabet $\Sigma=[\sigma]$. In this setting, computing the state reached by $A$ after reading a word $v$ of length $\ell$ takes $O(\ell \log \log |\Sigma|)$ (again: after reading each letter $v[i]$ of $v$, we find the transition that needs to be executed by searching the transition labelled with $v[i]$ in the $y$-fast trie storing all the transitions leaving the current state).
 
We further define an $n\times n$ boolean matrix $M$, with all entries initialized with false. Then, we simply run the DFA $A$ on the strings $w[i..n]$, for $i\in [n]$, and set $M_{ij}$ to true if and only if the automaton $A$ is in a final state after reading $w[i..j]$. The matrix $M$ can then be used to answer in $O(1)$-time queries asking whether $w[i..j]$ is in $L$ or not: we answer yes if and only if $M_{ij}$ is true. Clearly, the time needed to compute $M$ is $O(n^2\log \log |\Sigma|)\subseteq O(n^2\log \log n)$ (as $\Sigma=[\sigma]$ is an integer alphabet, with $\sigma\leq n$), to which we need to add $O(\size(A))$ time for reading the input and storing the DFA $A$.  
\end{proof}

\subsection{On the Pathwidth of Hamiltonian Graphs}
\begin{lemma}\label{lem:svnham}
The problem of deciding whether the vertex separation number is less than or equal to a given integer $k$ for graphs which have a Hamiltonian cycle (also given as input) is $\npclass$-complete.
\end{lemma}
\begin{proof}
Let $G=(V,E)$ be a graph with $V=\{v_1,\ldots,v_n\}$, and assume $\vsn(G)=k$. Moreover, let $\sigma=(v'_1,\ldots,v'_n)$ be an ordering of the vertices of $G$ such that the vertex separation number of the ordering $\sigma$ is $k$. 

Now, let $H=G \oplus K_n$ be the graph obtained by adding to $G$ the complete graph $K_n = (V_n = \{w_1, w_2, \ldots, w_n\}, E_n)$ and connecting all the vertices of $V_n$ to all the vertices of $G$. More precisely, we add the vertices $V_n$ to $V$ and edges connecting vertex $w_i$ to every vertex $x \in V\cup (W \setminus \{w_i\})$, for all $i\in [n]$. 

We will show that $\vsn(H)=n+\vsn(G)=n+k$. 

Intuitively, to compute the $\vsn$ of an ordering of the vertices (represented as a sequence of vertices) of some graph (in particular, of $H$ or $G$, in our case), we will traverse this sequence left to right with a pointer $p$. When $p$ reaches position $j$ of the sequence, we will see how many of the first $j-1$ vertices (called the left part of $\gamma$ w.r.t. the current position of the pointer) are connected to one of the last $n-j+1$ vertices (called the right part of $\gamma$ w.r.t. the current position of the pointer). 

So, let us first show that $H$ admits an ordering of its vertices with $\vsn$ equal to $n+k$. Indeed, this ordering is $\gamma=(w_1,\ldots,w_n,v'_1,\ldots v'_n)$. Clearly, for all $j\in [n]$, all vertices $w_1,\ldots,w_{j-1}$ are connected to all the nodes of the set $w_j,\ldots,w_n, v'_1,\ldots,v'_n$. Moreover, all the vertices $w_1,\ldots,w_{n}$ are connected to all vertices in $v'_1,\ldots,v'_n$. So, while traversing with the pointer $p$, left to right, the first $n+1$ positions of the sequence $\gamma$, the number of vertices from the left part w.r.t. the position of the pointer, which are connected to some vertex from the right part w.r.t. the position of the pointer, increases from $1$ to $n$. Consequently, while further traversing the positions of $\gamma$ the number of vertices from the left part w.r.t. the pointer, which are connected to nodes of the right part w.r.t. the pointer, stays always greater than or equal to $n$. In particular, when the pointer reaches $v'_j$ (that is, it already went over $w_1,\ldots,w_n,v'_1,\ldots v'_{j-1}$), the number of vertices from the left part w.r.t. the pointer, which are connected to nodes of the right part w.r.t. the pointer, is $n$ plus the number of vertices of $v'_1,\ldots v'_{j-1}$ which are connected to some vertex of $v'_j,\ldots,v'_n$. But, computed over $\ell\in[n]$, the maximum number of vertices of $v'_1,\ldots v'_{\ell-1}$ which are connected to vertices of $v'_{\ell},\ldots,v'_n$ is exactly $k$ (as the $\vsn$ of the ordering $\sigma$ is $k$). Therefore, the $\vsn$ of the ordering $\gamma$ is $n+k$. 

Now, assume, for the sake of a contradiction, that $H$ admits an ordering of its vertices which has $\vsn$ equal to $q<n+k$. We denote by $\alpha=(\alpha_1, v''_1, \alpha_2, v''_2, \ldots, \alpha_n,v''_n,\alpha_{n+1})$ the ordering of the vertices of $H$ with minimum $\vsn$, where $\alpha_i$ is a (potentially empty) sequence of vertices from the set $V_n$, for all $i\in [n]$. Again, we traverse this sequence with a pointer, and observe that the number of vertices from the left side w.r.t. the position of the pointer, which are connected to the nodes from the right side, is at least $n$. Let $h$ be such that $\alpha_{\ell}$ is empty for all $\ell>h$ and $\alpha_h$ is non-empty. Then, while traversing $\alpha$, until reaching the position of $v''_{h+1}$, the number of vertices from the left side w.r.t. the position of the pointer, which are connected to the nodes from the right side, grows from $1$ to $n+h$ (the latter value being reached when the left part consists of $(\alpha_1, v''_1, \alpha_2, v''_2, \ldots, \alpha_h,v''_h)$ and the right part of $(v''_{h+1},\ldots,v''_n)$. As the ordering $\gamma$ of the vertices of $H$ has $\vsn$ equal to $n+k$, and $\alpha$ is the ordering with minimum $\vsn$, it immediately follows that $h<k$. While continuing our traversal of $\alpha$, with the pointer moving over $(v''_{h+1},\ldots,v''_n)$, when reaching $v''_\ell$ the number of vertices from the left side w.r.t. the position of the pointer (i.e., $(\alpha_1, v''_1, \alpha_2, v''_2, \ldots, \alpha_h,v''_h,\ldots,v''_\ell)$), which are connected to the nodes from the right side (i.e., $(v''_{\ell+1},\ldots,v''_n)$), equals $n$ plus the number of vertices from $(v''_1,\ldots,v''_\ell)$ which are connected to vertices from $(v''_{\ell+1},\ldots, v''_n)$. Let $n+t $ be the maximum, over all $\ell \in \{h,h+1,\ldots,n\} $, number of vertices from the left side w.r.t. the position of the pointer (i.e., $(\alpha_1, v''_1, \alpha_2, v''_2, \ldots, \alpha_h,v''_h,\ldots,v''_\ell)$), which are connected to some nodes from the right side (i.e., $(v''_{\ell+1},\ldots,v''_n)$). Clearly, $n+t=q $, i.e., $n+t$ is the $\vsn$ of $\alpha$. Also, we get that $t\geq h$ and $t<k$ (as the $\vsn$ of $\alpha$ is $n+t$ and this is smaller than $n+k$, which is the $\vsn$ of $\gamma$). Further, we show that the ordering $\pi=(v''_1,\ldots,v''_n)$ of the vertices of $G$ has $\vsn$ is equal to $t$. Indeed, when moving the pointer over the first $h+1$ vertices in $\pi$, the number of vertices from the left side w.r.t. the position of the pointer (i.e., $(v''_1,\ldots,v''_\ell)$, where $\ell\leq h$), which are connected to nodes from the right side (i.e., $(v''_{\ell+1},\ldots,v''_n)$) stays at most $h$. Then, from the analysis of $\alpha$ done above, it follows that, when going with the pointer over the positions $\ell >h$, the maximum number of vertices from the left side w.r.t. the position of the pointer (i.e., $(v''_1,\ldots,v''_\ell)$), which are connected to nodes from the right side (i.e., $(v''_{\ell+1},\ldots,v''_n)$) is upper bounded by $t$ (and this value is actually reached for some position $\ell$). So, the $\vsn$ of $\pi$ is equal to $t<k$. But this is a contradiction with the assumption that $\vsn(G)=k$. 

In conclusion $\vsn(H)=n+k$. 

Now, coming to the statement of our problem, we proceed as follows. The following problem is known to be NP-complete \cite{Bodlaender98}: given a graph $G$ and an integer $k$, decide whether $\vsn(G)\leq k$.  We can now immediately reduce this problem to the following problem: given a Hamiltonian graph $G$, a Hamiltonian cycle of $G$, and an integer $k$, decide whether $\vsn(G)\leq k$. Indeed, for an input of the first problem (a graph $G$ and an integer $k$) we construct the input of the second problem as follows: the graph $G\oplus K_n$ (where $n$ is the number of vertices of $G$), the Hamiltonian cycle $(v_1,w_1,v_2,w_2,\ldots,v_n,w_n, v_1)$ (where $\{v_1,\ldots,v_n\}$ are the vertices of $G$ and $\{w_1,\ldots,w_n\}$ are the vertices of $K_n$), and the integer $n+k$. The proof above shows the correctness of this reduction ($\vsn(G)\leq k$ if and only if $\vsn(G\oplus K_n)\leq n+k$), which can clearly be performed in polynomial time. Therefore the second problem is also $\npclass$-complete, which completes the proof of our statement.\ 

\end{proof}

\section{Proof of Theorem~\ref{UpperBoundsTheorem}}

The theorem follows from the following two lemmas. 

\begin{lemma}\label{regConstkXPLemma}
$\matchProb$ with regular constraints can be solved in time $\bigO(|v| |\Sigma| |p| \gapsize{\mathcal{C}}^{|\mathcal{C}|})$.
\end{lemma}

\begin{proof}
For every $(i, j, L_{i, j}) \in \mathcal{C}$, let $A_{i, j}$ be the DFA that represents $L_{i, j}$, i.\,e., $\lang{A_{i, j}} = L_{i, j}$.\par
We transform the DFAs $A_{i, j}$ for every $(i, j, L_{i, j}) \in \mathcal{C}$ into an NFA $A$ as follows. The NFA $A$ accepts all strings $v$ with $p \subseq_e v$ as follows. We maintain a counter $i \in \{0, 1, \ldots, |p|\}$ initialised with $0$. Whenever the next input symbol equals $p[i + 1]$, $A$ can guess to either increment the counter to $i+1$ or to leave it unchanged, and $A$ accepts if the counter is $|p|$ after the input is completely read. Moreover, for every $(i, j, L_{i, j}) \in \mathcal{C}$, as soon as $A$ guessed that it just read symbol $p[i]$ (i.\,e., it increments the counter from $i-1$ to $i$), it starts simulating $A_{i, j}$ and it will change the counter to $j$ only if $A_{i, j}$ is in a final state at this moment (and rejects, if this is not the case). In this way, $A$ accepts a string $v$ with $p \subseq_e v$, and for every $(i, j, L_{i, j}) \in \mathcal{C}$, it also checks whether $\gap_{v, e}[i, j] \in \lang{A_{i, j}}$. Hence, $A$ accepts exactly the strings $v$ such that $p \subseq_e v$, and $v$ and $e$ satisfy all gap constraints of $\mathcal{C}$. \par
In order to estimate $A$'s size, we observe that its states must store the counter $i$ and a state of all $A_{j, j'}$ such that $(i, j, L_{i, j}) \in \mathcal{C}$ and $j \leq i < j'$. This means that its number of states can be bounded by $\bigO(|p| \prod_{(i, j, L_{i, j}) \in \mathcal{C}} |A_{i, j}|) = \bigO(|p| \gapsize{\mathcal{C}}^{|\mathcal{C}|})$, and, since the out-degree for each state is $\bigO(|\Sigma|)$, the total size of $A$ is bounded by $\bigO(|\Sigma| |p| \gapsize{\mathcal{C}}^{|\mathcal{C}|})$.\par
Hence, we can check whether there is an embedding $e$ with $p \subseq_{e} v$ such that $v$ and $e$ satisfy $\mathcal{C}$ by checking whether or not $v \in \lang{A}$. The latter can be checked in time $\bigO(|v| |A|) = \bigO(|v| |\Sigma| |p| \gapsize{\mathcal{C}}^{|\mathcal{C}|})$.
\end{proof}

\begin{lemma}\label{semilConstkXPLemma}
$\matchProb$ with semilinear length constraints can be solved in time $\bigO(|v|^{(2|\mathcal{C}| + 1)} + |v| |\mathcal{C}| \gapsize{\mathcal{C}})$.
\end{lemma}

\begin{proof}
Let $(v, p, \mathcal{C})$ be an instance of $\matchProb$ with $\mathcal{C} = \{C_1, C_2, \ldots, C_k\}$ and $C_\ell = (i_\ell, j_\ell, S_\ell)$ for every $\ell \in [k]$, and with $n = |v|$ and $m = |p|$. By Lemma~\ref{semilinearPreprocessLemma}, we can compute in time $\bigo(\sum^k_{\ell = 1}(n|S_\ell|)) = \bigo(n k \gapsize{\mathcal{C}})$ a data structure that, for every $\ell \in [k]$ and $x \in [n]$, allows us to answer whether $x \in S_\ell$ in constant time. \par
We enumerate every tuple
\begin{equation*}
((i'_{1}, j'_1), (i'_{2}, j'_2), \ldots, (i'_{k}, j'_k)) \in ([n] \times [n])^k\,,
\end{equation*}
such that, for every $\ell \in [k]$, $i'_{\ell} < j'_\ell$ and $(j'_\ell - i'_\ell - 1) \in S_{\ell}$. Note that there are $\bigO(n^{2k})$ such tuples and we can easily enumerate them with the help of the data structure from Lemma~\ref{semilinearPreprocessLemma}. \par
We note that each such tuple induces a partial mapping $e : [m] \to [n]$ by $i_\ell \mapsto i'_\ell$ and $j_\ell \mapsto j'_\ell$ for every $\ell \in [k]$. We will now check whether this partial mapping can be extended to an embedding $e' : [m] \to [n]$ that witnesses $p \subseq_{e'} v$ and that satisfies all gap-constraints.\par
Let $q_1, q_2, \ldots, q_{k'}$ be the elements of $[m]$ for which $e$ is defined. For every $j \in [k']$, we check whether $p[q_j] = v[e(q_j)]$ and, if yes, we set $e'(q_j) := e(q_j)$. Then, if possible, for every $j \in [k' - 1]$, we map the positions $q_j + 1, q_j + 2, \ldots, q_{j+1} - 1$ to positions in $\{e(q_j) + 1, \ldots, e(q_{j + 1}) - 1\}$ such that $e'(q_j + 1) < e'(q_j + 2) < \ldots < e'(q_{j+1} - 1)$ and, for every $r \in \{q_j + 1, q_j + 2, \ldots, q_{j+1} - 1\}$, $p[r] = v[e'(r)]$. Analogously, we map the positions $1, 2, \ldots, q_1 - 1$ to positions in $\{1, 2, \ldots, e(q_1) - 1\}$ such that $e'(1) < e'(2) < \ldots < e'(q_1- 1)$ and, for every $r \in \{1, 2, \ldots, q_1 - 1\}$, $p[r] = v[e'(r)]$, and we map the positions $q_{k'} + 1, q_{k'} + 2, \ldots, m$ to positions in $\{e'(q_{k'}) + 1, e'(q_{k'}) + 2, \ldots, n\}$ such that $e'(q_{k'} + 1) < e'(q_{k'} + 2) < \ldots < e'(m)$ and, for every $r \in \{q_{k'} + 1, q_{k'} + 2, \ldots, m\}$, $p[r] = v[e'(r)]$. All these mappings can be constructed greedily.\par
For a single tuple, this can be done in time $\bigO(m + n)$. Hence, the total running time is $\bigO(n k \gapsize{\mathcal{C}} + n^{2k + 1})$.
\end{proof}

\subsection{Proof of Theorem~\ref{cliqueTheorem}}

\begin{proof}
For the $k$-$\cliqueProb$ problem, we get an undirected graph $G = (V, E)$ and a number $k \in [|V|]$, and we want to decide whether there is a clique of size at least $k$, i.\,e., a set $K \subseteq V$ with $|K| \geq k$ and, for every $u, v \in K$ with $u \neq v$, we have that $\{u, v\} \in E$. It is a well-known fact that $k$-$\cliqueProb$ is $\wclass[1]$-hard. We will now define a parameterised reduction from $k$-$\cliqueProb$ to $\matchProb_{\semi}$ parameterised by $|p|$.

Let $(G, k)$ be an instance of $k$-$\cliqueProb$ and the graph $G$ consisting of vertices $v_1, \dots, v_n$ and represented by its adjacency matrix $A = (a_{ij})_{1 \leq i, j \leq n}$, where $a_{ij} = 1$ if $\{v_i, v_j\} \in E$ and $a_{ij} = 0$ otherwise. We can w.l.o.g.\ assume that each vertex of $G$ has a loop. Therefore $A$ is a symmetric $\{0, 1\}$-matrix with $1$'s on the diagonal. Now for $K \subseteq [n]$ with $|K| = k$ the vertices $\{v_i \mid i \in K\}$ are a $k$-clique of the graph if and only if
\[
    a_{ij} = 1 \qquad \forall\, i, j \in K
\]
Note that because of the assumed loops it is not necessary to demand $i \neq j$. It follows that $G$ has a $k$-clique, if and only if the $k \times k$ matrix containing only 1's is a principal submatrix of $A$, i.e., a submatrix where the set of remaining rows is the same as the set of remaining columns. By flattening out the matrices into strings, the problem of finding a principal submatrix can be framed as a problem of finding a subsequence, with constraints that ensure that such a subsequence actually corresponds to a principal submatrix. Therefore, we define the pattern 
\[
    p = 0\underbrace{1111 \cdots 1}_{k^2}0
\]
and the text
\[
    w = 0 a_{11} a_{12} \cdots a_{1n} a_{21} \cdots a_{2n} \cdots a_{n1} \cdots a_{nn} 0
\]

The extra zeros mark the start and end of the flattened matrices. Adding the constraint
\begin{equation}
    \label{eq_constraint_anchor}
    (1, k^2 + 2, L(n^2))
\end{equation}
ensures that the zeros of the pattern actually are embedded into the first and last character of the text, since these are the only positions that are far enough apart from each other. 

For more intuitive naming of the positions, let us now define 
\[
    (i, j)_m \colonequals 1 + m(i-1) + j
\]
for $m \in \mathbb{N}, i, j \in [m]$. Then $(i, j)_k$ and $(i, j)_n$ can be used to refer to the position of the entries in the $i$-th row and $j$-th column of the flattened matrices in $p$ or $w$ respectively, for example $w[(i, j)_n] = a_{ij}$ (recall the preceding $0$ in $p$ and $w$, which explains the preceding ``$1 +$'' in the definition of $(i, j)_m$). 
We can now formalize the reasoning from above:

\medskip

\noindent\emph{Claim 1}: $G$ has a $k$-clique if and only if for some $1 \leq s_1 < \dots < s_k \leq n$ the function $e: [k^2 + 2] \to [n^2 + 2]$ with $e(1) = 1$, $e(k^2 + 2) = n^2+2$ and
    \[
        e((i, j)_k) = (s_i, s_j)_n \qquad \forall\, i, j \in [k]
    \]
    is an embedding of $p$ into $w$.

\medskip

\noindent\emph{Proof of Claim $1$}: $e$ is strictly increasing, therefore, it is a valid embedding if and only if 
    \[
        a_{s_is_j} = w[(s_i, s_j)_n] = p[(i, j)_k] = 1 \qquad \forall\, i, j \in [k]
    \]
    This means that $v_{s_1}, \dots, v_{s_k}$ are pairwise adjacent and form a $k$-clique of $G$.  \qed(\emph{Claim 1})

\medskip

For an embedding $e$ of $p$ into $w$ that satisfies constraint (\ref{eq_constraint_anchor}) and thus has $1 = e(1) < e((i, j)_k) < e(k^2 + 2) = n^2 + 2$ for $i, j \in [k]$ we define $R^{(e)}_{ij}, C^{(e)}_{ij} \in [n]$ as
\[
    (R^{(e)}_{ij}, C^{(e)}_{ij})_n = e((i, j)_k) 
\]
i.e., the row and column of $A$ where the element from the $i$-th row and $j$-th column of the principal submatrix is embedded.

We now have to introduce additional constraints, such that, for any embedding $e$ that satisfies these constraints, there are $s_1, \dots, s_k \in [n]$ with $R^{(e)}_{ij} = s_i$ and $C^{(e)}_{ij} = s_j$ for all $i, j \in [k]$. Let us first make sure that the elements on the diagonal are also embedded into the diagonal, i.e., that the submatrix we find is indeed principal. For this, we make use of the fact that the first character of the pattern is embedded in the first position of the text. Thus, if we introduce the constraints 
\begin{equation}
    \label{eq_constraint_diagonal}
    (1, (i, i)_k, L(0; n+1)) \qquad \text{for $i \in [k]$}
\end{equation}
it must follow that 
\[
    (R^{(e)}_{ii}, C^{(e)}_{ii})_n = e((i, i)_k) = e(1) + m_i n = 1 + m_i n = (m_i+1, m_i+1)_n
\]
for some $m_i \in \mathbb{N}_0$. We get $s_i \colonequals R^{(e)}_{ii} = C^{(e)}_{ii} = m_i + 1$, i.e., the $i$-th diagonal element of $p$ gets matched to the $s_i$-th diagonal element of $w$.

Now, we will enforce that all elements from the same column are also embedded into the same column. For this, we can use that $(r, c)_n \equiv c + 1 \mod n$. Then, adding the constraints
\begin{equation}
    \label{eq_constraint_columns}
    ((i, j)_k, (i+1, j)_k, L(n-1; n)) \qquad \text{for $i \in [k-1], j \in [k]$}
\end{equation}
ensures that for any $j \in [k]$ we have $C^{(e)}_{1j} = C^{(e)}_{2j} = \ldots = C^{(e)}_{kj}$. Lastly, we make sure that elements from the same row all are embedded into the same row by adding the constraints
\begin{equation}
    \label{eq_constraint_row}
    ((i, 1)_k, (i, k)_k, \{0\} \cup [n-1]) \qquad \text{for $i \in [k]$}
\end{equation}
Note that $[n] = \bigcup_{i=1}^n L(i) = \{1, 2, \ldots, n\}$ is a semilinear set.

Now, for any $i, j, j' \in [k]$ with $j < j'$ it is $R^{(e)}_{ij} \leq R^{(e)}_{ij'}$ and 
\begin{align*}
    n &\geq e((i, k)_k) - e(i, 1)_k) \\ 
    &\geq e((i, j')_k) - e((i, j)_k) \\ 
    &= (R^{(e)}_{ij}, C^{(e)}_{ij})_n - (R^{(e)}_{ij'}, C^{(e)}_{ij'})_n\\
    &= (1 + n(R^{(e)}_{ij'} - 1) + C^{(e)}_{ij'} - (1 + n(R^{(e)}_{ij} - 1) + C^{(e)}_{ij} \\ 
    &= n(R^{(e)}_{ij'} - R^{(e)}_{ij}) + \underbrace{s_{j'} - s_j}_{> 0}
\end{align*}

Therefore $R^{(e)}_{ij} = R^{(e)}_{ij'}$ and the elements must all be embedded into the same row (namely $s_i$). 

Now, any embedding $e$ that embeds $p$ into $w$ and satisfies the constraints (\ref{eq_constraint_anchor}),~(\ref{eq_constraint_diagonal}),~(\ref{eq_constraint_columns})~and~(\ref{eq_constraint_row}) fulfills all the properties defined by \emph{Claim $1$} above. Likewise, it can easily be verified that every embedding $e$ with these properties does also satisfy the given constraints. 

We observe that the alphabet of this reduction is constant, $\Sigma = \{0, 1\}$. Moreover, $|p| = \bigO(k^2)$, which means that it is a parameterised reduction from $k$-$\cliqueProb$ to $|p|$-$\matchProb_{\semi}$. This proves that $\matchProb_{\semi}$ parameterised by $|p|$ is $\wclass[1]$-hard, even for constant $|\Sigma|$. However, the constraints $((i, 1)_k, (i, k)_k, \{0\} \cup [n-1])$, $i \in [k]$, have not constant size (although all other constraints have indeed constant size). We will therefore next  improve the reduction to a parameterised Turing reduction with only constraints of constant size.  

Let us call the $\matchProb_{\semi}$-instance described above the \emph{original instance}, whereas, for every $d \in \{0\} \cup [n-1]$, we denote the above $\matchProb_{\semi}$-instance with every constraint $((i, 1)_k, (i, k)_k, \{0\} \cup [n-1])$ replaced by constraint $((i, 1)_k, (i, k)_k, L(d))$ as the \emph{$d$-variant} of the instance. We claim that the original instance is positive if and only if there is some $d \in \{0\} \cup [n-1]$ such that the $d$-variant is a positive instance. 

If the original instance is positive, then there are $1 \leq s_1 < s_2 < \ldots < s_k \leq n$ such that $e: [k^2 + 2] \to [n^2 + 2]$ with $e(1) = 1$, $e(k^2 + 2) = n^2+2$ and $e((i, j)_k) = (s_i, s_j)_n$ for every $i, j \in [k]$ is an embedding of $p$ into $w$ that satisfies all the constraints. However, for every $i \in [k]$, we have that $e((i, k)_k) - e((i, 1)_k) = s_k - s_1$, which means that the embedding also satisfies all constraints $((i, 1)_k, (i, k)_k, L(d))$ with $d = s_k - s_1 - 1$. Hence, the $d$-variant is positive for $d = s_k - s_1 - 1$.

On the other hand, assume that some $d$-variant is satisfied. Hence, there are $1 \leq s_1 < s_2 < \ldots < s_k \leq n$ such that $e: [k^2 + 2] \to [n^2 + 2]$ with $e(1) = 1$, $e(k^2 + 2) = n^2+2$ and $e((i, j)_k) = (s_i, s_j)_n$ for every $i, j \in [k]$ is an embedding of $p$ into $w$ that satisfies all the constraints of the instance, including the constraints $((i, 1)_k, (i, k)_k, L(d))$, $i \in [k]$. Since $d \in \{0\} \cup [n-1]$, this means that all the constraints $((i, 1)_k, (i, k)_k, \{0\} \cup [n-1])$, $i \in [k]$, are satisfied as well. Thus, the original instance is positive.

Consequently, constructing all the $d$-variants of the original instance describes a parameterised Turing reduction from $\cliqueProb$ to $\matchProb_{\semi}$. Moreover, the alphabet is constant and every semilinear length constraint has constant size. This proves that $\matchProb_{\semi}$ parameterised by $|p|$ is $\wclass[1]$-hard, even for constant $\gapsize{\mathcal{C}}$ and constant $|\Sigma|$.

Next, we adapt the original reduction from above (without the adaption that yields a Turing reduction with only constant size constraints) into to a reduction from $k$-$\cliqueProb$ to $|p|$-$\matchProb_{\reg}$. This can be easily done by representing every semilinear length constraint as a regular constraint. More precisely, $\{v \in \Sigma^* \mid |v| \in L(n^2)\}$ can be represented by a DFA with $\bigo(n^2)$ states, and $\{v \in \Sigma^* \mid |v| \in L(0; n)\}$, $\{v \in \Sigma^* \mid |v| \in L(n-1; n-1)\}$ and $\{v \in \Sigma^* \mid |v| \in \{0\} \cup [n-1]\}$ can each be represented by a DFA with $\bigO(n)$ states.
 This describes a parameterised reduction from $k$-$\cliqueProb$ to $|p|$-$\matchProb_{\semi}$. The alphabet is still constant, but the constraints do not have constant size anymore. This proves that $\matchProb_{\reg}$ parameterised by $|p|$ is $\wclass[1]$-hard, even for constant $|\Sigma|$.

\end{proof}

\subsection{Proof of Theorem~\ref{SATTheorem}}

\begin{proof}
We reduce from monotone $1$-in-$3$-3SAT: Let $A = \{x_1, x_2, \ldots, x_n\}$ be a finite set and let $c_1, c_2, \ldots, c_m \subseteq A$ with $|c_i| = 3$ for every $1 \leq i \leq m$. Decide whether there is a subset $B \subseteq A$ such that $|c_i \cap B| = 1$ for every $1 \leq i \leq m$.\par
Let now $A$ and $c_1, c_2, \ldots, c_m$ be such an instance. For the sake of concreteness, we also set $c_j = \{x_{\ell_j, 1}, x_{\ell_j, 2}, x_{\ell_j, 3}\}$ for every $j \in [m]$, and with $x_{\ell_j, 1} < x_{\ell_j, 2} < x_{\ell_j, 3}$ for some order `$<$' on $A$. \par
We transform this instance into two strings $u_A$ and $v_A$ over the alphabet $\{\tb, \#\}$ as follows: 

\begin{align*}
&u_A& &=& &\overbrace{\tb \# \tb \# \ldots \tb \#}^{n \text{ times}}& &\overbrace{\tb \# \tb \# \ldots \tb \#}^{m \text{ times}}\,,& \\
&v_A& &=& &\underbrace{\tb^2 \# \tb^2 \# \ldots \tb^2 \#}_{n \text{ times}}& &\underbrace{\tb^3 \# \tb^3 \# \ldots \tb^3 \#}_{m \text{ times}}\,.& 
\end{align*}

For every $i \in [n + m]$, the position of $u_A$ and $v_A$ with the $i^{\text{th}}$ occurrence of $\#$ is called the \emph{$i^{\text{th}}$ separator} of $u_A$ and $v_A$, respectively. For every $i \in [n]$, we will call the $i^{\text{th}}$ occurrence of $\tb$ of $u_A$ the \emph{$x_i$-block} of $u_A$, and for every $j \in [m]$, we will call the $(n + j)^{\text{th}}$ occurrence of $\tb$ of $u_A$ the \emph{$c_j$-block} of $u_A$. We also denote the $\tb^2$ and the $\tb^3$ factors of $v_A$ as the $x_i$-blocks of $v_A$ and as the $c_j$-blocks of $v_A$ in an analogous way.\par
For every $i \in [n]$, we denote the position of $u_A$ that corresponds to the $x_i$-block of $u_A$ by $\pi_u(x_i)$, and for every $j \in [n]$, we denote the position of $u_A$ that corresponds to the $c_j$-block of $u_A$ by $\pi_u(c_j)$. We use analogous terminology for the positions in $x_i$-blocks or $c_j$-blocks of $v_A$, i.\,e., we use $\pi_v(x_i, p)$ for $i \in [n]$ and $p \in [2]$ to denote the $p^{\text{th}}$ position of the $x_i$-block of $v_A$, and we use $\pi_v(c_j, p)$ for $j \in [m]$ and $p \in [3]$ to denote the $p^{\text{th}}$ position of the $c_j$-block of $v_A$.\par
If $u_A \subseq_e v_A$ for some embedding $e : [|u_A|] \to [|v_A|]$, then $e$ must be such that every separator of $u_A$ is mapped to its corresponding separator of $v_A$. In particular, this implies that every $x_i$-block of $u_A$ is embedded as a subsequence into the $x_i$-block of $v_A$, and every $c_j$-block of $u_A$ is embedded as a subsequence into the $c_j$-block of $v_A$.\par
For every $i \in [n]$, there are exactly two different ways of how the $x_i$-block of $u_A$ can be embedded as a subsequence into the $x_i$-block of $v_A$: position $\pi_u(x_i)$ is mapped to either position $\pi_v(x_i, 1)$ of $v_A$ (i.\,e., the left occurrence of $\tb$ of the $x_i$-block of $v_A$) or to position $\pi_v(x_i, 2)$ of $v_A$ (i.\,e., the right occurrence of $\tb$ of the $x_i$-block of $v_A$). We will interpret the second way of how the $x_i$-block is mapped as ``selecting $x_i$ into the solution set $B$''. Let us make this precise by the following table (by $(\ast_1)$ we denote the kind of mapping that we interpret as selecting $x_i$ and by $(\ast_0)$ we denote the other one):\smallskip

\begin{tabular}{c|cccc|cccc}
$x_i$-block in  & \multicolumn{4}{c|}{$(\ast_0)$} & \multicolumn{4}{c}{$(\ast_1)$} \\\hline
$u_A$ & $\#$ & $\tb$ &               & $\#$ & $\#$ &               & $\tb$ & $\#$ \\\hline
$v_A$ & $\#$ & $\tb$ & $\tb$ & $\#$ & $\#$ & $\tb$ & $\tb$ & $\#$
\end{tabular}\smallskip

We conclude that we can therefore interpret an embedding $e$ with $u_A \subseq_e v_A$ as a subset $B_e \subseteq A$ that contains exactly those $x_i$ such that the $x_i$-block of $u_A$ is embedded according to $(\ast_1)$. The next goal is to add gap-constraints that force this set $B_e$ to be a solution set, i.\,e., $|B_e \cap c_j| = 1$ for every $j \in [m]$. For this, the $c_j$-blocks will be important.\par
For every $j \in [m]$, there are exactly three different ways of how the $c_j$-block of $u_A$ can be embedded as a subsequence into the $c_j$-block of $v_A$ (recall that every $c_j$-block of $u_A$ is embedded as a subsequence into the $c_j$-block of $v_A$), depending on whether $\pi_u(c_j)$ is mapped to the first or the second or the third occurrence of $\tb$ of $v_A$'s $c_j$-block. Let us make this again precise with a table (we denoted the three possible embeddings by $(\dagger_1)$, $(\dagger_2)$ and $(\dagger_3)$, respectively):\smallskip

\begin{tabular}{c|ccccc|ccccc|ccccc}
$c_j$-block in & \multicolumn{5}{c|}{$(\dagger_1)$} & \multicolumn{5}{c|}{$(\dagger_2)$} & \multicolumn{5}{c}{$(\dagger_3)$}\\\hline
$u_{A}$ & $\#$ & $\tb$ &         &          & $\#$ &  $\#$  &         & $\tb$ &          & $\#$ & $\#$ &          &         & $\tb$ & $\#$           \\\hline
$v_{A}$ & $\#$ & $\tb$ & $\tb$ & $\tb$ & $\#$ & $\#$ & $\tb$ & $\tb$ & $\tb$ & $\#$ & $\#$ & $\tb$ & $\tb$ & $\tb$ & $\#$
\end{tabular}\smallskip

Let $p \in [3]$ and let $e$ be an embedding with $u_A \subseq_e v_A$. If $e$ embeds the $c_j$-block of $u_A$ according to $(\dagger_{p})$, then we interpret this as the situation that the element $x_{\ell_j, p}$ of $c_j$ (i.\,e., the $p^{\text{th}}$ element of $c_j$) is the unique element in $B \cap c_j$. Consequently, $e$ can also be interpreted as selecting exactly one element from each $c_j$ in this way. What we need to achieve next is a kind of synchronisation in the sense that $x_i \in B$ (meaning that the $x_i$-block is embedded according to $(\ast_1)$) if and only if all $c_j$-blocks with $x_{\ell_j, p} = x_i$ are also embedded according to $(\dagger_p)$. More precisely, we will next add gap-constraints such that the following property holds. 

\begin{quote}
Property $(\ddagger)$: Let $e$ be an embedding with $u_A \subseq_e v_A$ that satisfies the gap-constraints, let $i \in [n]$, $j \in [m]$ and $p \in [3]$ such that $x_i = x_{\ell_j, p}$.
\begin{itemize}
\item The $x_i$-block of $u_A$ is embedded according to $(\ast_1)$ if and only if the $c_j$-block of $u_A$ is embedded according to $(\dagger_p)$.
\item The $x_i$-block of $u_A$ is embedded according to $(\ast_0)$ if and only if the $c_j$-block of $u_A$ is embedded according to $(\dagger_{q})$ with $q \in [3] \setminus \{p\}$.
\end{itemize}
\end{quote}

Before defining the gap-constraints, we first observe that with such gap-constrains (i.\,e., gap-constraints that yield property $(\ddagger)$), we have that $A$ and $c_1, c_2, \ldots, c_m$ is a positive $1$-in-$3$-3SAT-instance if and only if $u_A \subseq_{\mathcal{C}} v_A$, where $\mathcal{C}$ is the set of all the gap-constraints (yet to be defined). Indeed, if $u_A \subseq_{\mathcal{C}} v_A$, then the embeddings of the $x_i$-blocks determine some $B_e \subseteq A$, and property $(\ddagger)$ directly implies that $|B \cap c_j| = 1$ for every $j \in [m]$ (note also that the three ways of embedding $c_j$-blocks $(\dagger_1)$, $(\dagger_2)$ and $(\dagger_3)$ are obviously mutually exclusive). Furthermore, if $B$ is a solution set for the $1$-in-$3$-3SAT-instance $A$ and $c_1, c_2, \ldots, c_m$, then the mapping $e$ induced by embedding every $x_i$-block according to $(\ast_1)$ or $(\ast_0)$ depending on whether or not $x_i \in B$, and embedding every $c_j$-block according to $(\dagger_1)$ or $(\dagger_2)$ or $(\dagger_3)$ depending on whether $x_{\ell_{j, 1}} \in B$ or $x_{\ell_{j, 2}} \in B$ or $x_{\ell_{j, 2}} \in B$ (since $B$ is a solution set, exactly one of these situations arises) will satisfy all gap-constraints and $u_A \subseq_{e} v_A$.\par

We now define the gap-constraints. To this end, let $i \in [n]$ and $j \in [m]$ be such that $x_i \in c_j$. There are three possible cases.
\begin{align*}
&\text{Case }x_{\ell_j, 1} = x_i:& &C_{\pi_u(x_i),\pi_u(c_j)} =& &\{\# w \#,& &\tb w \# \tb,& &\tb w \# \tb \tb& &\mid w \in \{\tb, \#\}^*\}\,.&\\
&\text{Case }x_{\ell_j, 2} = x_i:& &C_{\pi_u(x_i),\pi_u(c_j)} =& &\{\tb w \#,& &\# w \# \tb,& &\tb w \# \tb \tb& &\mid w \in \{\tb, \#\}^*\}\,.&\\
&\text{Case }x_{\ell_j, 2} = x_i:& &C_{\pi_u(x_i),\pi_u(c_j)} =& &\{\tb w \#,& &\tb w \# \tb,& &\# w \# \tb \tb& &\mid w \in \{\tb, \#\}^*\}\,.&
\end{align*}

We will now prove that these gap-constraints imply the property $(\ddagger)$ mentioned above. To this end, let $e$ be an embedding with $u_A \subseq_e v_A$ that satisfies the gap-constraints, let $i \in [n]$, $j \in [m]$ and $p \in [3]$ such that $x_i = x_{\ell_j, p}$. There are again three cases to consider:
\begin{itemize}
\item Case $x_i = x_{\ell_j, 1}$: Let us first assume that the $x_i$-block of $u_A$ is embedded according to $(\ast_1)$. This means that the $(\pi_u(x_i), \pi_u(c_j))$-gap must start with symbol $\#$, which according to gap-constraint $C_{\pi_u(x_i),\pi_u(c_j)}$, means that the $(\pi_u(x_i),\pi_u(c_j))$-gap has the form $\# w \#$ for some $w \in \{\tb, \#\}^*$. This is only possible, if $e$ maps position $\pi_u(c_j)$ of $u_A$ (i.\,e., the $c_j$-block of $u_A$) to position $\pi_v(c_j, 1)$ of $v_A$ (i.\,e., the first position of the $c_j$-block of $v_A$). This directly implies that the $c_j$-block of $u_A$ is embedded according to $(\dagger_1)$.\par
Let us next assume that the $c_j$-block of $u_A$ is embedded according to $(\dagger_1)$. This means that $e$ maps position $\pi_u(c_j)$ of $u_A$ to position $\pi_v(c_j, 1)$ of $v_A$. Thus, the $(\pi_u(x_i),\pi_u(c_j))$-gap must end with an occurrence of $\#$, which, according to gap-constraint $C_{\pi_u(x_i),\pi_u(c_j)}$, means that the $(\pi_u(x_i),\pi_u(c_j))$-gap has the form $\# w \#$ for some $w \in \{\tb, \#\}^*$. Consequently, the $x_i$-block of $u_A$ is embedded according to $(\ast_1)$.\par
Next, assume that the $x_i$-block of $u_A$ is embedded according to $(\ast_0)$. This means that the $(\pi_u(x_i), \pi_u(c_j))$-gap must start with symbol $\tb$, which according to gap-constraint $C_{\pi_u(x_i),\pi_u(c_j)}$, means that the $(\pi_u(x_i),\pi_u(c_j))$-gap has the form $\tb w \# \tb$ or $\tb w \# \tb \tb$ for some $w \in \{\tb, \#\}^*$. The first case is only possible, if $e$ maps position $\pi_u(c_j)$ of $u_A$ to position $\pi_v(c_j, 2)$ of $v_A$ (i.\,e., the second position of the $c_j$-block of $v_A$), and the second case is only possible, if $e$ maps position $\pi_u(c_j)$ of $u_A$ to position $\pi_v(c_j, 3)$ of $v_A$ (i.\,e., the third position of the $c_j$-block of $v_A$). Consequently, the $c_j$-block of $u_A$ is embedded according to $(\dagger_2)$ or $(\dagger_3)$, respectively.\par
Finally, let us assume that the $c_j$-block of $u_A$ is embedded according to either $(\dagger_2)$ or $(\dagger_3)$. This means that $e$ maps position $\pi_u(c_j)$ of $u_A$ to either position $\pi_v(c_j, 2)$ or $\pi_v(c_j, 3)$ of $v_A$. In the first case, the $(\pi_u(x_i),\pi_u(c_j))$-gap must end with $\# \tb$, which, according to gap-constraint $C_{\pi_u(x_i),\pi_u(c_j)}$, means that the $(\pi_u(x_i),\pi_u(c_j))$-gap has the form $\tb w \# \tb$ for some $w \in \{\tb, \#\}^*$. Hence, the $x_i$-block of $u_A$ must be  embedded according to $(\ast_0)$. In the second case, the $(\pi_u(x_i),\pi_u(c_j))$-gap must end with $\# \tb \tb$, which, according to gap-constraint $C_{\pi_u(x_i),\pi_u(c_j)}$, means that the $(\pi_u(x_i),\pi_u(c_j))$-gap has the form $\tb w \# \tb \tb$ for some $w \in \{\tb, \#\}^*$. Again, this implies that the $x_i$-block of $u_A$ must be embedded according to $(\ast_0)$. 
\item Case $x_i = x_{\ell_j, 2}$: This case can be handled analogously, so we can be a bit more concise. If the $x_i$-block of $u_A$ is embedded according to $(\ast_1)$, then, due to gap-constraint $C_{\pi_u(x_i),\pi_u(c_j)}$, the $(\pi_u(x_i),\pi_u(c_j))$-gap must have the form $\# w \# \tb$ for some $w \in \{\tb, \#\}^*$. This is only possible, if $e$ maps position $\pi_u(c_j)$ of $u_A$ to position $\pi_v(c_j, 2)$ of $v_A$, which implies that the $c_j$-block of $u_A$ is embedded according to $(\dagger_2)$. On the other hand, if the $c_j$-block of $u_A$ is embedded according to $(\dagger_2)$, then the $(\pi_u(x_i),\pi_u(c_j))$-gap must end with $\# \tb$, thus, the gap-constraint $C_{\pi_u(x_i),\pi_u(c_j)}$ implies that the $x_i$-block of $u_A$ can only be embedded according to $(\ast_1)$.\par
If the $x_i$-block of $u_A$ is embedded according to $(\ast_0)$, then the $(\pi_u(x_i),\pi_u(c_j))$-gap must have the form $\tb w \#$ or $\tb w \# \tb \tb$ for some $w \in \{\tb, \#\}^*$, which is only possible if the $c_j$-block of $u_A$ is embedded according to $(\dagger_1)$ or $(\dagger_3)$, respectively. If the $c_j$-block of $u_A$ is embedded according to either $(\dagger_1)$ or $(\dagger_3)$, then the $(\pi_u(x_i),\pi_u(c_j))$-gap must end with $\#$ or $\# \tb \tb$. Again the gap-constraint $C_{\pi_u(x_i),\pi_u(c_j)}$ directly implies that the $(\pi_u(x_i),\pi_u(c_j))$-gap has the form $\tb w \#$ or $\tb w \# \tb \tb$ for some $w \in \{\tb, \#\}^*$, which, in both cases, means that the $x_i$-block of $u_A$ must be  embedded according to $(\ast_0)$. 
\item Case $x_i = x_{\ell_j, 3}$: This case can be shown analogously.
\end{itemize}

Finally, it can be easily verified that the gap-constraints are regular languages that can be represented by $\DFA$s with at most $8$ states.

\end{proof}

\section{Proof of \cref{thm:constant-vsn}}\label{sec:appendixVSN}

\begin{proof}
We will start by showing the lower bound. This follows immediately by observing that the reduction supporting the claim in Theorem \ref{cliqueTheorem} is a parameterized reduction, mapping the parameter $k$ of the  $k-\cliqueProb$ to the parameter $k$ for the  $\matchProb_{{\mathcal L}, {\mathcal V}_k}$ problem. Indeed, the pattern $p$ constructed in the respective reduction has length $2+k^2$, so the $\vsn$ of its graph representation to the gap constraints constructed for $p$ is upper bounded by $1+k^2$. In fact, as an additional observation, the graph representation of the respective set of gap constraints contains a $(k \times k)$-grid as a minor. Therefore, the pathwidth of that graph and, as such, its $\vsn$ is also lower bounded by $k$. 

For the upper bound, we proceed as follows.

Consider an input for the $\matchProb_{{\mathcal L}, {\mathcal V}_k}$ problem: two strings $w, p \in \Sigma^*$, with $|w|=n$ and $|p|=m\leq n$, and a non-empty set $\mathcal{C}$ of gap-constraints such that the graph representation of ${\mathcal C}$ has $\vsn$ at most $k$. The case $m=1$ is trivial, so we will consider $m\geq 2$. Assume ${\mathcal C}=\{C_1,\ldots, C_s\}$, where $C_i=(a_i,b_i,L_i)$, for $i\in [s]$. Set $G_{\mathcal C}=([m], E_{\mathcal C})$ be the graph representing ${\mathcal C}$ (as introduced in Section \ref{sec:vsn-and-intervals}); it is immediate that, since $m\geq 2$, we have that $\vsn(G_{\mathcal C})\geq 1$. Let ${\mathcal D}=\{C_1,\ldots, C_s\}$, where $C_i=(a_i,b_i,L_i)$, for $i\in [s]$, are the labels of the edges of $G_{\mathcal C}$; that is ${\mathcal D}$ extends ${\mathcal C}$ with the trivial constraints added during the construction of the graph $G_{\mathcal C}$. 

We begin with a preprocessing phase. This is in turn also split into two sub-phases. In the first one, we only analyse the set of constraints ${\mathcal C}$ and the corresponding graph $G_{\mathcal C}$. 

We first compute an ordering $\sigma=(i_1,\ldots,i_m)$ of the vertices of $G_{\mathcal C}$ (i.e., of $[m]$), which has $\vsn$ at most $k$. By the results from \cite{Bodlaender96,Bodlaender98}, we get that $\sigma$ can be computed in linear time $O(m)$, where the constant hidden by the $O$-notation is depending superexponentially on $k$. 

Now, by simply traversing the ordering $\sigma$ and analysing the graph, we can construct, for each $j\in [m]$, the set $W_j=\{i_\ell \mid \ell \leq j,$ and there exists $k> j$ and an edge $(i_\ell,i_ k)\in E_C\}$. The size of $W_j$ is upper bounded by $k=\vsn(G_{\mathcal C})$, and can be computed trivially in $O(m^3)$ (more efficient methods exist, but there are other parts of this algorithm which are more time consuming, so we will not insist on them). It is important to note here that, for $j\in [2,m]$, $W_{j}$ can be obtained from $W_{j-1}$ by potentially removing some of the elements of $W_j$ and potentially adding $i_j$ to the resulting set; moreover, $W_1=\{i_1\}$.

In the second sub-phase, we preprocess the word $w$. In the case of $\semi$-constraints, by Lemma \ref{semilinearPreprocessLemma}, we can build data structures allowing us to decide whether $x \in L_j$ for some $x \in [n]$ in constant time, after an $O(n\size(L_j))\subset O(n^2)$ time processing (by Remark \ref{sizeSemiLin}), for $j\in [s]$. We need to do this  process for each of the constraints $C_1,\ldots, C_s$, so the overall time is $O(n^2m^2)$, as $s\leq m^2$. 

In the case of $\reg$-constraints, by Lemma \ref{regularPreprocessLemma}, we can build data structures allowing us to decide whether $w[i..j] \in L_h$ for all $i,j\in [n]$ in constant time and $h\in [s]$, after time $O(n^2\log \log \log n + \size(L_h))$. We need to do this  process for each of the constraints $C_1,\ldots, C_s$, so the overall time is $O(n^2s \log \log n +\size({\mathcal C}))\subseteq O(n^2m^2 \log \log n)$ (as $\size({\mathcal C})\leq m^2n^2)$ by Remark \ref{sizeSemiLin} and the fact that $|{\mathcal C}|\leq m^2$). 

In conclusion, no matter in which case we are (that is, $\matchProb_{{\reg}, {\mathcal V}_k}$ or $\matchProb_{{\semi}, {\mathcal V}_k}$), after a polynomial time preprocessing we can answer queries ``is $w[i..j]\in L_h$?'' in $O(1)$ time, for all $i,j\in [n], h\in [s]$.

Once the preprocessing phase is completed, we can proceed with the main matching algorithm.

This will proceed in $m+1$ phases. 

In phase $1$, we construct the set $U_1=\{(1,i)\mid w[i]=p[i_1]\}$ (all embeddings of $p[i_1]$ in $w$); note that in this phase, no constraints need to be verified, but node $i_1$ is connected to at least another vertex, as $m\geq 2$). 

We assume that after phase $q$, the set $U_q$ contains all tuples $(q,j_1,\ldots, j_{t_q})$, such that:
\begin{itemize}
\item $j_1<j_2<\ldots<j_{t_q}$;
\item there exists an embedding of $p[i_1], p[i_2], \ldots, p[i_q]$ in $w$ which verifies all constraints of ${\mathcal D}$ of the form $(a,b,L)$ with $a,b\in \{i_1,\ldots, i_{q}\}$, $a<b$, and, moreover, if $W_{q}=\{a_1,\ldots,a_{t_q}\}$ then $p[a_h]$ is embedded in the symbol of $w$ on the position $j_h$, for $h\in [t_q]$.
\end{itemize}
Basically, $U_q$ keeps track, for all correct embeddings of $p[i_1], p[i_2], \ldots, p[i_q]$ in $w$, of the positions of $w$ where the symbols of $p$, which are affected by constraints which were not already satisfied, are embedded. $U_1$ fulfils these conditions.
 
Further, we proceed by dynamic programming. For $q\geq 1$, in phase $q+1$ we construct the set $U_{q+1}$ based on set $U_q$ as follows. Let $W_{q}=\{a_1,\ldots,a_{t_q}\}$ and $W_{q+1}=\{b_1,\ldots,b_{t_{q+1}}\}$ (and assume that $a_i<a_j$ and $b_i<b_j$, for $i<j$). 

For each $(q,j_1,\ldots, j_{t_q})\in U_q$, there exists an embedding $f$ of $p[i_1], p[i_2], \ldots, p[i_q]$ in $w$ which verifies all constraints of ${\mathcal D}$ of the form $(a,b,L)$ with $a,b\in \{i_1,\ldots, i_{q}\}$, and, moreover, if $W_{q}=\{a_1,\ldots,a_{t_q}\}$ then $p[a_h]$ is embedded in the symbol of $w$ on the position $j_h$, for $h\in [t_q]$. Now, for each $j$ such that $w[j]=p[i_{q+1}]$, we check whether there exists any constraint $(a_\ell,i_{q+1},L)$ (respectively, $(i_{q+1},a_\ell,L)$)  such that $j_\ell < j$ and $w[j_\ell..j]\notin L$ (respectively, $j_\ell > j$ and $w[j..j_\ell]\notin L$). If this is not the case, then we decide that $f$ can be extended by embedding $p[i_{q+1}]$ into $w[j]$. Therefore, we add the tuple $(q+1,j'_1,\ldots, j'_{t_{q+1}})$ to the set $U_{q+1}$ where $w[j'_\ell]$ is the embedding of $p[b_\ell]$ under $f$ (note that the tuple $(q+1,j'_1,\ldots, j'_{t_{q+1}})$ can be computed in $O(|W_{q}|)$ time, given the way $W_{q+1}$ is obtained from $W_{q}$, already mentioned above). To avoid inserting duplicate elements into $U_{q+1}$, we maintain a characteristic array for this set; that is, we maintain a boolean $t_{q+1}$-dimensional array $M_{q+1}$, where $M_{q+1}[j'_1,\ldots, j'_{t_{q+1}}]$ is true if and only if $(q+1,j'_1,\ldots, j'_{t_{q+1}})$ was already inserted in $U_{q+1}$. 

The construction of $U_{q+1}$ is correct (in the sense that it fulfils the desired properties) as every correct embedding of of $p[i_1], p[i_2], \ldots, p[i_{q+1}]$ in $w$ must extend a correct embedding of $p[i_1], p[i_2], \ldots, p[i_q]$ in $w$, and we do such an extension only if all constraints affecting $i_{q+1}$ are fulfilled. In particular, note here that the constraints of the form $(i,i+1,\cdot)$ enforce that no two positions of $p$ are mapped to the same position of $w$. As far as the complexity is concerned, this extension can be clearly done in $O(n^{k+1}m)$ time, as $|U_q|\leq n^{k}$, and for each $j\in [n]$ (in the worst case), we need to check at most $k\leq m$ constraints and construct a tuple of size at most $k+1\leq m+1$; maintaining the array $M_{q+1}$ can be done in $O(n^k)$ time.  

In the end of our algorithm, we simply need to check whether $U_m$ contains any element. If yes, then we conclude that $p$ is a ${\mathcal C}$-subsequence of $w$.

The overall complexity of this algorithm is $O(n^{k+1}m^2 )$ to which the preprocessing time is added. The statement of the theorem follows.
\end{proof}

\section{Non-Intersecting Constraints: Upper Bound (Proof of \cref{thm:NonIntersectUB}) } \label{sec:NonIntersectUB}

In this section, we develop a recursive algorithm for solving $\matchProb_{{\mathcal L},{\mathcal H}}$ for ${\mathcal L}\in \{\semi,\reg\}$, i.e., the variant of $\matchProb$ for pairwise non-intersecting semilinear or regular constraints constraints. 

In this problem, let $w \in \Sigma^n$ and $p \in \Sigma^m$ (with $m\leq n$) be the input strings and $\mathcal{C} = \{C_1, \dots, C_K\}$ with $C_k = (a_k, b_k, L_k)$ for $k \in [K]$ be the set of constraints. Recall, that we want to check whether $p$ is a ${\mathcal C}$-subsequence of $w$.

\subsection{Preprocessing Constraints}

In the case of $\semi$-constraints, by Lemma \ref{semilinearPreprocessLemma}, we can build data structures allowing us to decide whether $x \in L_k$ for some $x \in [n]$ in constant time, after an $O(n\size(L_k))$ time processing. We need to do this  process for each of the constraints $C_1,\ldots, C_K$, so the overall time is $O(n\size({\mathcal C}))$. By Remark \ref{sizeSemiLin}, which states that we work under the assumption that $\size(|L_k|)\in O(n)$, building all these data structures takes at most $O(n^2K)$ time.

In the case of $\reg$-constraints, by Lemma \ref{regularPreprocessLemma}, we can build data structures allowing us to decide whether $w[i..j] \in L_k$ for all $i,j\in [n]$ in constant time, after an $O(n^2\log \log n + \size(L_k))$. We need to do this  process for each of the constraints $C_1,\ldots, C_K$, so the overall time is $O(n^2K \log \log n +\size({\mathcal C}))$. By Remark \ref{sizeSemiLin}, which states that we work under the assumption that $\size(|L_k|)\in O(n^2)$, building all these data structures takes at most $O(n^2K\log \log n)$ time.

In conclusion, no matter in which case we are (that is, $\matchProb_{{\reg},{\mathcal H}}$ or $\matchProb_{{\semi},{\mathcal H}}$), after an $O(n^2K\log \log n)$-time preprocessing we can answer queries "is $w[i..j]\in L_k$?" in $O(1)$ time, for all $i,j\in [n], k\in [K]$.

For simplicity, in the following we only deal with $\semi$-constraints, as the $\reg$-constraints can be treated similarly. Firstly, in the next two subsections we define the basic ingredients of our dynamic programming machinery, and then we will explain how they are used to get the algorithm supporting the statement of Theorem \ref{thm:NonIntersectUB}.

\subsection{Embedding Substrings}

In order to solve the given problem, we will deconstruct it into smaller problems (as customary for recursive algorithms). For this, we consider the embedding of a substring $p[s..t]$ for $s, t \in [m]$ with $s \leq t$. We denote the set of relevant constraints in the interval $[s, t]$ with
\[
	\restr{\mathcal{C}}{[s, t]} \colonequals \{C_k \in \mathcal{C} \mid s \leq a_k < b_k \leq t\}
\]
and $\mathcal{C}^{(k)} \colonequals \restr{\mathcal{C}}{[a_k, b_k]}$.
Then, we can ask whether (and where) $p[s..t]$ can be embedded into $w$:

\begin{definition}
	For $s, t \in [m]$ with $s \leq t$ and $\mathcal{C'} \subset \restr{\mathcal{C}}{[s, t]}$ we define the boolean $n \times n$-matrix $P(s, t, \mathcal{C}')$: The entry $(P(s, t, \mathcal{C}'))_{ij}$ is true, if and only if there is a partial embedding $e: [s, t] \to [n]$ with $i = e(s)$, $j = e(t)$ and $p[s..t] \preceq_e w$ (see \autoref{fig_embedding_substrings}) that satisfies all the constraints in $\mathcal{C}'$. Here, $P(s, t, \mathcal{C}')$ is called a \emph{partial embeddings matrix}.
\end{definition}

\begin{figure}[ht]
	\centering
	\begin{tikzpicture}
		\node (P) {$p = {}$};
		\node (W) [below of = P] {$w = $};
		\node (W1) [right= of W] {\texttt{z}};
		\node (W2) [right of = W1] {\texttt{x}};
		\node (W3) [right of = W2] {\texttt{y}};
		\node (W4) [right of = W3] {\texttt{x}};
		\node (W5) [right of = W4] {$\dots$};
		\node (W6) [right of = W5] {\texttt{x}};
		\node (W7) [right of = W6] {\texttt{z}};
		\node (W8) [right of = W7] {\texttt{y}};
		\node (W9) [right of = W8] {\texttt{x}};
		\node (P1) [above of = W3]{\texttt{x}};
		\node (P2) [right of = P1]{\texttt{y}};
		\node (P3) [right of = P2]{$\dots$};
		\node (P4) [right of = P3]{\texttt{z}};
		\node (P5) [right of = P4]{\texttt{x}};

		\node (S) [above of = P2] {\scriptsize{$s$}};
		\node (T) [above of = P4] {\scriptsize{$t$}};
		\node (I) [below of = W3] {\scriptsize{$i$}};
		\node (J) [below of = W7] {\scriptsize{$j$}};

		\draw[->] (S) -- (P2);
		\draw[->] (T) -- (P4);
		\draw[->] (I) -- (W3);
		\draw[->] (J) -- (W7);

		\fill[color=red!20] (W3.north) -- (W7.north) -- (P4.south) -- (P2.south);
		\draw (P2.south) -- (W3.north);
		\draw (P4.south) -- (W7.north);
		\path (P3) -- node[midway] {$e$} (W5);
	\end{tikzpicture}
	\caption{Embedding $p[s..t]$ into $w[i..j]$}\label{fig_embedding_substrings}
\end{figure}
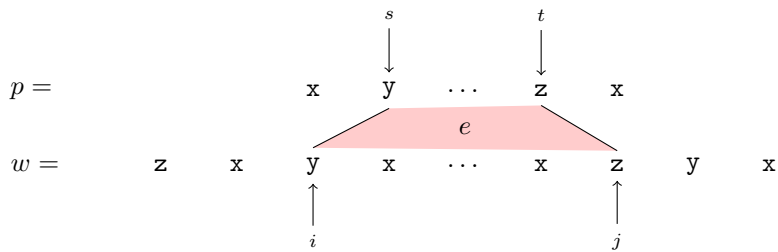

Let us first consider partial embeddings matrices without constraints:
\[
	B^{(s, t)} \colonequals P(s, t, \emptyset)
\]
If we find one unconstrained partial embedding $e: [s, t] \to [n]$, it is very easy to construct other partial embeddings by increasing $e(t)$ (or decreasing $e(s)$) whilst keeping everything else the same. We will use this to efficiently calculate $B^{(s, t)}$. 

Hereby, we fix $i \in [n]$ and then construct the smallest possible embedding $e$, with respect to $e(t)$, starting at position $i$ (or at some position after $i$). The value of $e(t)$ is then
\[
	bound_i(s, t) \colonequals \min \{j \in [i, n] \mid p[s..t] \preceq w[i..j]\}
\]
with $\min\{\} \colonequals \infty$. The value of $bound_i(s, t)$ can be efficiently computed using the following properties:
\begin{lemma}\label{lem_bound}
	For $i \in [n], s, t \in [m], s < t$, it holds:
	\begin{enumerate}
		\item $bound_i(s, s) = \min\{j \in [n] \mid j \geq i, w[j] = p[s]\}$
		\item $bound_i(s, t) = \min\{j \in [n] \mid j > bound_i(s, t-1), w[j] = p[t]\}$
	\end{enumerate}
\end{lemma}

\begin{proof}~
	\begin{enumerate}
		\item Let $j \colonequals \min\{j \in [n] \mid j \geq i, w[j] = p[s]\}$. We have $p[s] \preceq_e w[i..j]$ with $e(s) = j$ and $p[s]$ cannot be a subsequence of $w[i..j-1]$ since the character $p[s]$ is not contained in $w[i..j-1]$.
		\item Let $j \colonequals \min\{j \in [n] \mid j > bound_i(s, t-1), w[j] = p[t]\}$ and $j' \colonequals bound_i(s, t-1)$. By definition we have an embedding $e'$ with $p[s..t-1] \preceq_{e'} w[i..j']$. It is $i \leq e'(s) < \dots < e'(t-1) \leq j' < j$. Then $e: [s, t] \to [n]$ with $e(u) = \begin{cases} e'(u) & u < t \\ j & u = t \end{cases}$ is an embedding of $p[s..t]$ into $w[i..j]$.

		      Let us now assume that $j$ is not minimal, i.e., there exists an embedding $e$ with $p[s, t] \preceq_e w[i..j-1]$. $e$ must satisfy $w[e(t)] = p(t)$ with $e(t) < j$. If $e(t) > j'$, this contradicts the minimality of $j$. If $e(t) \leq j'$ then $e(t-1) < e(t) \leq j'$ and thus $e|_{[s, t-1]}$ is an embedding of $p[s, t-1]$ into $w[i..j'-1]$. This contradicts the minimality of $bound_i(s, t-1)$. \qedhere
	\end{enumerate}
\end{proof}

\autoref{algo_bound} uses these properties to calculate $bound_i(s, t)$. To reduce the number of parameters to the function and increase readability, some variables (in this case $w$ and $p$) are injected.

\begin{algorithm}[ht]
	\caption{Calculating $bound_i(s, t)$}\label{algo_bound}
	\begin{algorithmic}
		\Function{bound}{$i$, $s$, $t$}
		\State $j \gets i-1$
		\State $t' \gets s$
		\While{$t' \leq t$}
		\Repeat
		\State $j \gets j + 1$
		\If{$j > n$}
		\State \Return $\infty$
		\EndIf
		\Until{$w[j] = p[t']$}
		\Comment{After this loop exits, it is $j = bound_i(s, t')$}
		\State $t' \gets t' + 1$
		\EndWhile
		\State \Return $j$
		\EndFunction
	\end{algorithmic}
\end{algorithm}

The algorithm requires $\mathcal{O}(n)$ time and, naturally, it closely resembles the algorithm that was sketched out in the introduction, which also tries to find an (unconstrained) subsequence as early as possible (i.e., is greedy). Note that this algorithm implicitly constructs the sequence $bound_i(s, s), \dots, bound_i(s, t)$, which corresponds to the values $e(s), \dots, e(t)$ of the embedding $e$ that the other algorithm would have found.

Now with $bound_i(s, t)$ calculated, it is possible to compute $B^{(s, t)}$:

\begin{lemma}\label{lem_B}
	For $s, t \in [m], s \leq t$ and $i, j \in [n]$ we have
	\[
		B^{(s, t)}_{ij} \iff w[i] = p[s] \;\land\; bound_i(s, t) \leq j \;\land\; w[j] = p[t]
	\]
\end{lemma}

\begin{proof}
	If $B^{(s, t)}$ is true, then there exists an embedding $e: [s, t] \to [n]$ with $p[s..t] \preceq_e w[i..j]$, $i = e(s)$ and $j = e(t)$. By definition of $bound$ we have $bound_i(s, t) \leq j$. Also, $w[i] = w[e(s)] = p[s]$ and $w[j] = w[e(t)] = p[t]$.

	If $bound_i(s, t) \leq j$, then there exists an embedding $e: [s, t] \to [n]$ with $i \leq e(s)$ and $e(t) \leq j$ with $p[s..t] \preceq_e w[i..j]$. If also $w[i] = p[s]$ and $w[j] = p[t]$, then we can define the function $e': [s, t] \to [n]$ as
	\[
		e'(u) = \begin{cases} i & u = s \\ e(u) & s < u < t \\ j & u = t \end{cases}
	\]
	$e'$ is strictly increasing, $w[e'(s)] = w[i] = p[s]$, $w[e'(t)] = w[j] = p[t]$ and $w[e'(u)] = w[e(u)] = p[u]$ for all $u$ with $s < u < t$. Thus $p[s..t] \preceq_{e'} w[i..j]$ and $B^{(s, t)}_{ij}$ must be true.
\end{proof}

Therefore, computing $B^{(s, t)}$ can be done by calculating $bound_i(s, t)$ for all $i \in [n]$ and then checking the conditions of the previous lemma for each pair $(i, j) \in [n] \times [n]$, which takes $\mathcal{O}(n^2)$ time:

\begin{algorithm}[ht]
	\caption{Computing $B^{(s, t)}$}\label{algo_Bst}
	\begin{algorithmic}
		\Function{Bst}{$s$, $t$}
		\For{$i \gets 1, n$}
		\State $b \gets$ \Call{bound}{$i$, $s$, $t$}
		\For{$j \gets 1, n$}
        \State $B_{ij} \gets (w[i] = p[s] \land b \leq j \land w[j] = p[t])$
		\EndFor
		\EndFor
		\State \Return $B$
		\EndFunction
	\end{algorithmic}
\end{algorithm}

Now that we know how to compute unconstrained partial embeddings matrices, let us see how to add constraints. Since checking the constraint $C_k$ involves verifying that $e(b_k) - e(a_k) \in L_k$, we have to know the values of $e(b_k)$ and $e(a_k)$, which we only do for partial embeddings matrices of the form $Q = P(a_k, b_k, \mathcal{C}')$. Now, considering an entry $Q_{ij}$ of this matrix, every corresponding embedding $e$ has $e(b_k) - e(a_k) = j - i$. Thus, either all of these embeddings satisfy $C_k$ (if $j - i -1 \in L_k$) and the entry is not changed when adding $C_k$, or none of them do (if $j - i -1 \notin L_k$) and the entry must be false after adding $C_k$. We therefore obtain the following lemma:

\begin{lemma}[Constraining partial embeddings]
	\label{lem_constraining_embeddings}
	For $k \in [K]$ and $\mathcal{C}' \subset \mathcal{C}^{(k)}$ it is
	\[
		P(a_k, b_k, \mathcal{C}' \cup \{C_k\}) = P(a_k, b_k, \mathcal{C}') \land mask(C_k)
	\]
	with $\land$ being applied elementwise and
	\[
		mask(C_k) \colonequals (j - i -1 \in L_k)_{1 \leq i,j \leq n}
	\]
	as the \emph{mask of constraint $C_k$}.
\end{lemma}

Lastly, we will see how to join embeddings:

\begin{lemma}[Joining partial embeddings]
	\label{lem_joining_embeddings}
	For $s, t, u \in [m]$ with $s < t < u$, $\mathcal{C}' \subset \restr{\mathcal{C}}{[s, t]}$ and $\mathcal{C}'' \subset \restr{\mathcal{C}}{[t, u]}$ it is
	\[
		P(s, u, \mathcal{C}' \cup \mathcal{C}'') = P(s, t, \mathcal{C}') \cdot P(t, u, \mathcal{C}'')
	\]
	with $(\cdot)$ being boolean matrix multiplication.
\end{lemma}

\begin{proof}
	Let $V = (V_{ij}) = P(s, t, \mathcal{C}')$, $W = (W_{ij}) = P(t, u, \mathcal{C}'')$ and $Q = (Q_{ij}) = P(s, u, \mathcal{C}' \cup \mathcal{C}'')$. Thus, we need to show that for any $i, j \in [n]$
	\[
		Q_{ij} \iff \bigvee_{k \in [n]} V_{ik} \land W_{kj}
	\]
	holds.

	If $Q_{ij}$ is true, then there exists some partial embedding $e: [s, u] \to [n]$ with $p[s..u] \preceq_e w$ that has $e(i) = s$, $e(j) = u$ and satisfies all constraints in $\mathcal{C}' \cup \mathcal{C}''$. But now $V_{ik}$ and $W_{kj}$ with $k \colonequals e(t)$ must also be true, because $\restr{e}{[s, t]}$ and $\restr{e}{[t, u]}$ inherently satisfy all required conditions and constraints.

	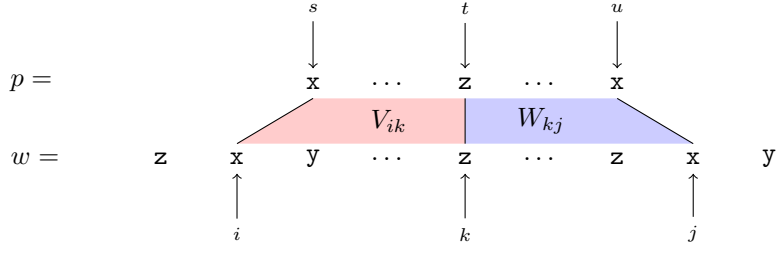
\begin{figure}[ht]
		\centering
		\begin{tikzpicture}
			\node (P) {$p = {}$};
			\node (W) [below of = P] {$w = $};
			\node (W1) [right= of W] {\texttt{z}};
			\node (W2) [right of = W1] {\texttt{x}};
			\node (W3) [right of = W2] {\texttt{y}};
			\node (W4) [right of = W3] {$\dots$};
			\node (W5) [right of = W4] {\texttt{z}};
			\node (W6) [right of = W5] {$\dots$};
			\node (W7) [right of = W6] {\texttt{z}};
			\node (W8) [right of = W7] {\texttt{x}};
			\node (W9) [right of = W8] {\texttt{y}};
			\node (P1) [above of = W3]{\texttt{x}};
			\node (P2) [right of = P1]{$\dots$};
			\node (P3) [right of = P2]{\texttt{z}};
			\node (P4) [right of = P3]{$\dots$};
			\node (P5) [right of = P4]{\texttt{x}};

			\node (S) [above of = P1] {\scriptsize{$s$}};
			\node (T) [above of = P3] {\scriptsize{$t$}};
			\node (U) [above of = P5] {\scriptsize{$u$}};
			\node (I) [below of = W2] {\scriptsize{$i$}};
			\node (K) [below of = W5] {\scriptsize{$k$}};
			\node (J) [below of = W8] {\scriptsize{$j$}};

			\draw[->] (S) -- (P1);
			\draw[->] (T) -- (P3);
			\draw[->] (U) -- (P5);
			\draw[->] (I) -- (W2);
			\draw[->] (K) -- (W5);
			\draw[->] (J) -- (W8);

			\fill[color=red!20] (W2.north) -- (W5.north) -- (P3.south) -- (P1.south);
			\fill[color=blue!20] (W5.north) -- (W8.north) -- (P5.south) -- (P3.south);
			\draw (P1.south) -- (W2.north);
			\draw (P3.south) -- (W5.north);
			\draw (P5.south) -- (W8.north);
			\path (P2) -- node[midway] {$V_{ik}$} (W4);
			\path (P4) -- node[midway] {$W_{kj}$} (W6);
		\end{tikzpicture}
		\caption{Joining partial embeddings}\label{fig_joining_embeddings}
	\end{figure}

	Conversely, if $V_{ik}$ and $W_{kj}$ are true for some $k \in [n]$, then there exist corresponding partial embeddings $e_1: [s, t] \to [n]$ and $e_2: [t, u] \to [n]$ which can be combined to some embedding $e: [s, u] \to [n]$ with
	\[
		e(t') \colonequals \begin{cases} e_1(t') & t' \in [s, t] \\ e_2(t') & t' \in [t, u] \end{cases} \qquad \text{for all $t' \in [s, u]$}
	\]
    Note that $e$ is well-defined, because $e_1(t) = k = e_2(t)$, and has $e(s) = e_1(s) = i$ and $e(u) = e_2(u) = j$. Furthermore, $e$ satisfies every constraint $C_k \in \mathcal{C}'$, since $a_k, b_k \in [s, t]$ and thus $e(b_k) - e(a_k) = e_1(b_k) - e_1(a_k) \in L_k$. By the same argument $e$ also satisfies all the constraints in $\mathcal{C}''$ and $Q_{ij}$ must be true.
\end{proof}

Being able to join partial embeddings matrices with a simple boolean matrix multiplication is not only elegant, but also allows to take advantage of fast algorithms for multiplying matrices \cite{FastMatrixMult}, \cite{FasterMatrixMult}. Therefore, joining embedded partial matrices can be done in $\mathcal{O}(n^\omega)$ time (with $\omega \approx 2.37$).

\subsection{Structuring Constraints}

Previously, we have seen how to compute unconstrained partial embeddings matrices, add constraints to them and join them together. It remains to be seen in which order these operations must be applied to obtain the solution to the problem.

We have seen that adding a constraint $C_k$ is only possible if the matrix is concerned with the substring $p[a_k..b_k]$. The other constraints in $\mathcal{C}^{(k)}$ should already be included, since there is no way of adding them later. Therefore, for any $k \in [K]$, at some point we have to calculate
\[
	A^{(k)} \colonequals P(a_k, b_k, \mathcal{C}^{(k)})
\]
Let us now recall the original problem. We need to decide whether it is possible to embed $p$ into $w$ while satisfying all constraints. In the context of our subproblems, this is solved by the matrix $P(1, m, \mathcal{C})$. More precisely, if and only if any of the entries of this matrix are true, then the answer to our problem is yes. Now since $\mathcal{C} = \restr{\mathcal{C}}{[1, m]}$, if we have a constraint $C_t = (1, m, L_t)$, then $A^{(t)}$ is exactly the matrix we are looking for. But even if this is not the case, we can just introduce a new constraint $(1, m, L(1;1))$. Since $L(1;1) = \mathbb{N}$, we are not disallowing any distances and the constraint does not change our solution. Let us generalize these two cases:

\begin{definition}
	The root index of $\mathcal{C}$ is defined as
	\[
		root(\mathcal{C}) \colonequals \begin{cases}
			k & \exists\, k \in [K]: a_k = 1 \land b_k = m \\
			0 & \text{otherwise}
		\end{cases}
	\]
\end{definition}

For the rest of the section we will use $r = root(\mathcal{C})$. In the case $r = 0$ we define an additional constraint $C_0 = (a_0, b_0, L_0) \colonequals (1, m, [n])$ and generalize all previous definitions to also include the case $k = 0$.

\begin{algorithm}[ht]
	\caption{Finding $root(\mathcal{C})$}
	\begin{algorithmic}
		\Function{root}{$\mathcal{C}$}
		\Comment{$\mathcal{C} = \{C_1, \dots, C_K\}$}
		\For{$k \gets 1,m$}
		\State $(a, b, L) \gets C_k$
		\If{$a = 1 \land b = m$}
		\State \Return $k$
		\EndIf
		\EndFor
		\State \Return 0
		\EndFunction
	\end{algorithmic}
\end{algorithm}

Before we can calculate $A^{(k)}$, we must first add all constraints in $\mathcal{C}^{(k)} \setminus \{C_k\}$, which means that we have to calculate $A^{(k')}$ for all $k'$ with $C_{k'} \in \mathcal{C}^{(k)} \setminus \{C_k\}$. Let us investigate this set more carefully. We have
\begin{align*}
	\mathcal{C}^{(k)} & = \{C_{k'} \mid a_k \leq a_{k'} < b_{k'} \leq b_k\}        \\
	                  & = \{C_{k'} \mid \interval(C_{k'}) \subseteq \interval(C_k)\} \\
	                  & = \{C_{k'} \mid C_{k'} \sqsubset C_k\} \cup \{C_k\}
\end{align*}

Therefore, we can say that $A^{(k)}$ (at least indirectly) depends on $A^{(k')}$, if $C_{k'} \sqsubset C_k$. For intuition let us consider the Hasse diagram of the poset $(\mathcal{C} \cup \{C_r\}, \sqsubset)$, which can be seen as a dependency graph. An example for this can be seen in \autoref{fig_constraint_graph}. The constraints are again drawn as arcs, with the dashed arc representing the newly introduced constraint $C_0$. The Hasse diagram is overlayed in red.

\begin{figure}[ht]
	\centering
	\begin{tikzpicture}[tree_node/.style={fill=red, midway, circle, scale=0.4},text_node/.style={text height=.6em}]
		\node[text_node] (P) {$p = {}$};
		\node[text_node] (P1) [right= of P]{\texttt{x}};
		\node[text_node] (P2) [right of = P1]{\texttt{y}};
		\node[text_node] (P3) [right of = P2]{\texttt{z}};
		\node[text_node] (P4) [right of = P3]{\texttt{y}};
		\node[text_node] (P5) [right of = P4]{\texttt{x}};
		\node[text_node] (P6) [right of = P5]{\texttt{x}};
		\node[text_node] (P7) [right of = P6]{\texttt{z}};
		\node[text_node] (P8) [right of = P7]{\texttt{y}};
		\node[text_node] (P9) [right of = P8]{\texttt{x}};

		\draw (P1.north) to[bend left = 60] node[tree_node] (C1) {} (P7.north);
		\draw (P7.north) to[bend left = 60] node[tree_node] (C2) {} (P9.north);
		\draw (P1.north) to[bend left = 60] node[tree_node] (C3) {} (P5.north);
		\draw (P5.north) to[bend left = 60] node[tree_node] (C4) {} (P7.north);
		\draw (P1.north) to[bend left = 60] node[tree_node] (C5) {} (P2.north);
		\draw (P2.north) to[bend left = 60] node[tree_node] (C6) {} (P3.north);
		\draw (P4.north) to[bend left = 60] node[tree_node] (C7) {} (P5.north);
		\draw (P5.north) to[bend left = 60] node[tree_node] (C8) {} (P6.north);
		\draw (P8.north) to[bend left = 60] node[tree_node] (C9) {} (P9.north);

		\draw[red, thick] (C1) -- (C3);
		\draw[red, thick] (C1) -- (C4);
		\draw[red, thick] (C2) -- (C9);
		\draw[red, thick] (C3) -- (C5);
		\draw[red, thick] (C3) -- (C6);
		\draw[red, thick] (C3) -- (C7);
		\draw[red, thick] (C4) -- (C8);

		\draw[dashed] (P1.north) to[bend left = 60] node[tree_node] (C0) {} (P9.north);
		\draw[red, thick] (C0) -- (C1);
		\draw[red, thick] (C0) -- (C2);
	\end{tikzpicture}
	\caption{Hasse diagram of constraints}\label{fig_constraint_graph}
\end{figure}
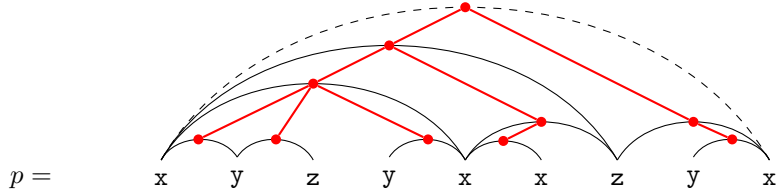

We immediately notice that the above diagram is an in-tree, i.e., a directed rooted tree with all edges pointing towards the root (in this case $C_r$). We can show that this is true in general by proving that there is exactly one path from each node to the root: We have $C_k \sqsubset C_r$ for all $k \in [K] \setminus \{r\}$, thus, by \autoref{lem_hasse}, there is a path from every node to the root. Also, assume that we had $C_k \sqsubsetdot C_p$ and $C_k \sqsubsetdot C_{p'}$ for pairwise distinct $k, p, p' \in [K] \cup \{r\}$, then $\interval(C_p) \cap \interval(C_{p'}) \supseteq \interval(C_k) \neq \emptyset$. Since $C_p$ and $C_{p'}$ cannot intersect, we have either $C_p \sqsubset C_{p'}$ or $C_{p'} \sqsubset C_p$ which contradicts $C_k \sqsubsetdot C_{p'}$ or $C_k \sqsubsetdot C_p$ respectively. From this, it follows that every constraint is covered by at most one other constraint and, thus, every node has at most one outgoing edge. Now, if we consider two (not necessarily distinct) paths from some node $C_k$ to $C_r$ and their longest common prefix ends on the node $C_{k'}$, then they cannot both continue in different directions because there is only one outgoing edge. Therefore, we must have $k' = r$ and since $C_r$ has no outgoing edge (no constraint contains $C_r$) both paths must end here and be identical.

If we consider two siblings $C_k$ and $C_{k'}$ with parent $C_p$, we see that both of them are covered by $C_p$. Thus, by an argument similar to the one above, they cannot be comparable w.r.t.\ $\sqsubset$. Therefore, their intervals do not intersect and must be comparable w.r.t.\ the interval order. Using this, the Hasse diagram can be seen as an ordered tree, where the children of any node are ordered. This ordering also automatically occurs in \autoref{fig_constraint_graph}.

Now let us consider pre-order Depth-first search on this tree starting from the root and the corresponding relation $<_T$, where $C_k <_T C_{k'}$ if $C_k$ is visited (not necessarily immediately) before $C_{k'}$. We now want to calculate this relation without relying on the covering relation: For two constraints $C_k \neq C_{k'}$ we have two different cases. If they are comparable w.r.t.\ $\sqsubset$, then $C_k <_T C_{k'}$ is true if and only if $C_{k'} \sqsubset C_k$, because we are traversing in pre-order. In this case, we obtain
\[
	(a_k < a_{k'} \land b_{k'} \leq b_k) \quad \lor \quad (a_k \leq a_{k'} \land b_{k'} < b_k)
\]
If they are not comparable, then we can consider the paths from $C_k$ and $C_{k'}$ to the root. Let $C_p$ and $C_{p'}$ be the elements of the paths preceeding the lowest common ancestor of $C_k$ and $C_{k'}$. Now $\interval(C_p) \supseteq \interval(C_k)$ (this even works if $k = p$) and $\interval(C_{p'}) \supseteq \interval(C_{k'})$. It is $C_k <_T C_{k'}$, if and only if $C_p <_T C_{p'}$ and therefore $\interval(C_p) < \interval(C_{p'})$ (note that $C_p$ and $C_{p'}$ are siblings). Then we have
\[
	b_k \leq b_p \leq a_{p'} \leq a_{k'}
\]
and, since $a_k < b_k$ and $a_{k'} < b_{k'}$, this implies
\[
	a_k < a_{k'} \land b_k < b_{k'}
\]
The two statements are indeed equivalent: In the second one, the constraints cannot be comparable w.r.t.\ $\sqsubset$ and therefore, $\interval(C_k) \cap \interval(C_{k'}) = [a_{k'}, b_k-1]$ must be empty, i.e., $b_k \leq a_{k'}$. Combining the two cases yields
\begin{align*}
    C_k <_T C_{k'} &\iff (a_k < a_{k'} \land b_{k'} > b_k) \lor (a_k < a_{k'} \land b_{k'} \leq b_k) \lor (a_k \leq a_{k'} \land b_{k'} < b_k) \\
    &\iff a_k < a_{k'} \lor (a_k \leq a_{k'} \land b_{k'} < b_k)
\end{align*}

With the preorder of the constraints calculated, constructing the tree can now be done as follows:
We visit all the nodes in the order given by $<_T$, starting with node $C_r$. If we are currently at some node $C_k$ and $C_{k'}$ is the next constraint in the preorder, then there are two possibilities: If $C_{k'} \sqsubset C_k$, then we can add $C_{k'}$ as the next child of $C_{k}$ and move to $C_{k'}$. Otherwise, $C_k$ does not have any more children and we can move to the parent of $C_k$ and try again. Note that $C_k <_T C_{k'}$ implies $a_k \leq a_{k'}$ and thus we only have to check $a_{k'} < b_k$ to distinguish between these two cases. \autoref{algo_children} implements this idea to calculate
\[
    children(k) \colonequals \{k' \mid C_{k'} \sqsubsetdot C_k\}
\]

\begin{algorithm}[ht]
	\caption{Determining $children(k)$ for all $k \in [K] \cup \{r\}$}\label{algo_children}
	\begin{algorithmic}
		\Function{children}{$\mathcal{C}$, $r$}
		\Comment{$\mathcal{C} = \{C_1, \dots, C_K\}$, $r = root(\mathcal{C})$}
		\State $C \gets [K] \setminus \{r\}$
		\State Sort $C$ with respect to $<_T$
		\For{$k \in [K] \cup \{r\}$}
		\State $c[k] \gets []$
		\Comment{$c$ is an array of lists}
		\EndFor
        \State $S \gets$ \Call{NewEmptyStack}{{}}
        \State \Call{push}{$S$, $r$}
		\For{$k' \in C$}
		\Comment{Iterate over $C$ in order}
		\Repeat
		\State $k \gets$ \Call{pop}{$S$}
		\Comment{Moving up the tree}
        \Until{$a_{k'} < b_k$}
		\Comment{Until $C_{k'} \sqsubset C_k$}
		\State \Call{push}{$S$, $k$}
		\Comment{Restore $k$}
		\State \Call{push}{$S$, $k'$}
		\Comment{Add $k'$ to the tree (as a child of $k$)}
		\State Append $k'$ to the end of $c[k]$
		\Comment{Put $k'$ as a child of $k$}
		\EndFor
		\State \Return $c$
		\EndFunction
	\end{algorithmic}
\end{algorithm}

Hereby, the stack maintains the path from the root to the current node and pushing and popping to and from the stack corresponds to moving down or up the tree.

With the tree now calculated, let us see how it is useful in computing $A^{(k)}$. As noted before, $A^{(k)}$ depends on $A^{(k')}$ if $C_{k'} \sqsubset C_k$. It therefore seems reasonable to view the tree as a dependency graph, where, to calculate $A^{(k)}$, we can use exactly the matrices $A^{(k')}$, where $C_{k'}$ is a child of $C_k$ in the tree. This can be done as follows:

\begin{lemma}
	Let $k \in [K]$, $d = |children(k)|$ . If $d = 0$, it is
	\[
		A^{(k)} = B^{(a_k, b_k)} \land mask(C_k)
	\]
	Otherwise we can order $children(k)$ as $\{k_1, \dots, k_d\}$ with $b_{k_v} \leq a_{k_{v+1}}$ for all $v \in [d-1]$ and
	\[
		A^{(k)} = \left(B^{(a_k, a_{k_1})} \cdot A^{(k_1)} \cdot B^{(b_{k_1}, a_{k_2})} \cdot A^{(k_2)} \cdots A^{(k_d)} \cdot B^{(b_{k_d}, b_k)}\right) \land mask(C_k)
	\]
\end{lemma}

\begin{proof}
	In the first case, the expression on the right is $P(a_k, b_k, \{C_k\})$ by \autoref{lem_constraining_embeddings}. Since $d = 0$, $C_k$ is a leaf of the Hasse diagram. Therefore $\{C_{k'} \mid C_{k'} \sqsubset C_k\}$ is empty and $A^{(k)} = P(a_k, b_k, \{C_k\})$.

	In the second case, the matrix product on the right calculates
	\[
		P\left(a_k, b_k, \bigcup_{v=1}^d \mathcal{C}^{(k_v)}\right)
	\]
	by \autoref{lem_joining_embeddings}. Now $\mathcal{C}^{(k_v)}$ contains exactly all the nodes from the subtree rooted in $C_{k_v}$. Therefore, after applying the mask (i.e., adding $C_k$) we have all constraints from the subtree rooted in $C_k$. This is exactly the constraints in $\mathcal{C}^{(k)}$ and therefore, the right side is equal to $A^{(k)}$.
\end{proof}

Using this lemma leads to the following algorithm (note that if $c$ is calculated via \textproc{children}, the constraints are already in the correct order):

\begin{algorithm}[ht]
	\caption{Calculating $A^{(k)}$}
	\begin{algorithmic}
		\Function{Ak}{$k$}
		\State $v \gets |c[k]|$
		\Comment{$c$ was calculated beforehand using \Call{children}{$\mathcal{C}$}}
		\If{$v = 0$}
		\State $A \gets$ \Call{Bst}{$a_k$, $b_k$}
		\Else
		\State $A \gets$ \Call{Bst}{$a_k$, $a_{c[k][1]}$}
		\For{$d \gets 1, v-1$}
		\State $A \gets A \cdot {}$\Call{Ak}{$c[k][d]$}
		\State $A \gets A \cdot {}$\Call{Bst}{$b_{c[k][d]}, a_{c[k][d+1]}$}
		\EndFor
		\State $A \gets A \cdot {}$\Call{Ak}{$c[k][v]$}
		\State $A \gets A \cdot {}$\Call{Bst}{$b_{c[k][v]}, b_k$}
		\EndIf
		\For{$i \gets 1, n$}
		\For{$j \gets 1, n$}
		\If{$j - i \notin L_k$}
		\State $A_{ij} \gets \texttt{false}$
		\EndIf
		\EndFor
		\EndFor
		\State \Return $A$
		\EndFunction
	\end{algorithmic}
\end{algorithm}

With this, \autoref{algo_problem} solves \textsc{Match} with pairwise non-intersecting length constraints (assuming preprocessed constraints).

\begin{algorithm}
    \caption{Solving \textsc{Match}}\label{algo_problem}
	\begin{algorithmic}
		\Function{Match}{$w$, $p$, $\mathcal{C}$}
		\State $r \gets$ \Call{root}{$\mathcal{C}$}
		\State $c \gets$ \Call{children}{$r$, $\mathcal{C}$}
		\State $A \gets$ \Call{Ak}{$r$}
		\Comment{Calculating $A^{(r)}$}
		\For{$i \gets 1, n$}
		\For{$j \gets 1, n$}
		\If{$A_{ij}$}
		\State \Return \texttt{yes}
		\EndIf
		\EndFor
		\EndFor
		\State \Return \texttt{no}
		\EndFunction
	\end{algorithmic}
\end{algorithm}

\subsection{Complexity Analysis}

The preprocessing phase, takes $O(n^2K)$ (for $\semi$-constrains). Then, in structuring the constraints, finding the root is done in $\mathcal{O}(K)$, the children can be calculated in $\mathcal{O}(K + n)$ time (the majority of time is needed for sorting the constraints w.r.t.\ the $<_T $ order, and this can be done by radix sort). When calculating the matrix $A^{(r)}$, most of the time is spent on matrix multiplications. For computing $A^{(k)}$ we need $2 \cdot |children(k)|$ of them (ignoring those in recursive calls). Since every $k' \in [K] \setminus \{r\}$ is contained in $children(k)$ for exactly one value of $k$ and calculating $A^{(r)}$ involves calling $A^{(k)}$ for all $k \in [K] \cup \{r\}$, the total number of matrix multiplications is $2 \cdot |[K] \setminus \{r\}| \in \mathcal{O}(K)$. Checking the entries of $A^{(r)}$ is then again easy and can be done in $\mathcal{O}(n^2)$ time. 
Therefore, we have shown the statement of Theorem \ref{thm:NonIntersectUB}, for $\semi$-constraints:
$\matchProb_{\semi, {\mathcal H}}$ can be solved in time $O(n^\omega K)$, where $K$ is the number of constraints in the input set of constraints ${\mathcal C}$.

For $\reg$-constraints, we have that the preoprocessing is done in $O(n^2K\log \log n)$. The rest of the algorithm is identical, and we obtain that $\matchProb_{\reg, {\mathcal H}}$ can be solved in time $O(n^\omega K + n^2\log \log n K)$, where $K$ is the number of constraints in the input set of constraints ${\mathcal C}$. Given the current state of the art regarding fast matrix multiplication, we conclude that  $\matchProb_{\reg, {\mathcal H}}$ can be solved in time $O(n^\omega K)$ 

The number of constraints $K$ is also upper bound by $O(n)$, as the graph representation of ${\mathcal C}$ is an outerplanar graph, with an edge for each constraint, and the number of edges in such a graph is linear in the number of vertices. Thus, we also get the following:

\begin{theorem}\label{thm:NonIntersectUBwithoutK}
For ${\mathcal L}\in\{\reg,\semi\}$, then $\matchProb_{{\mathcal L}, {\mathcal H}}$ can be solved in time $O(n^{\omega+1})$.
\end{theorem}

\section{Non-Intersecting Constraints: Lower Bound (Proof of \cref{thm:NonIntersectLB}) } \label{sec:NonIntersectLB}

This proof is given for $\semi$-constraints only.

We will now use a fine-grained complexity reduction from $3$-$\OV$ to prove a conditional lower bound for the problem of finding subsequences with non-intersecting length constraints. In particular, we are given three sets $A = \{\vec{a}_1, \dots, \vec{a}_n\}$, $B = \{\vec{b}_1, \dots, \vec{b}_n\}$ and $C = \{\vec{c}_1, \dots, \vec{c}_n\}$ with elements from $\{0, 1\}^d$, i.e., $d$-dimensional boolean vectors, and want to determine whether there are $i^*, j^*, k^* \in [n]$, such that $\sum_{\ell=1}^d\vec{a}_{i^*}[\ell] \cdot \vec{b}_{j^*}[\ell] \cdot \vec{c}_{k^*}[\ell]=0$. This should be achieved by encoding our input into a pattern $p$, a text $w$ and a set of constraints $\mathcal{C}$, such that the answer to the $3$-$\OV$ problem is yes, if and only if $p$ is a $\mathcal{C}$-subsequence of $w$.

Reducing $k$-$\OV$ problems to variants of the subsequence problem is nothing new, as this has already been done in \cite{DayEtAl2022} and \cite{BoundedSubsequences}. In the previous cases though, the authors used $2$-$\OV$ instead, which only has two sets of vectors. Thus, it is possible to encode one of the sets in the pattern and the other one in the text such that finding a subsequence corresponds to finding a pair of orthogonal vectors. Here, due to the higher complexity of the problem, using $2$-$\OV$ is not an option. $3$-$\OV$, on the other hand, requires to encode \emph{three} sets of vectors in only the pattern and the text.
Therefore, we will have to divide our pattern and text into two parts each and then use one of the parts to compare $A$ and $B$ and the other to compare $A$ and $C$. Using some constraints, we can then connect both parts and check whether an orthogonal triple of vectors exists.

For $i \in [n], u \in [d]$ let $a_{iu}$, $b_{iu}$ and $c_{iu}$ be the $u$-th component of $\vec{a}_i$, $\vec{b}_i$ and $\vec{c}_i$ respectively. We will now use the characters $\{0, 1, \#, @, \$, \S\}$ to encode the problem:
First, we define bit encodings $\mathsf{C}_p$ and $\mathsf{C}_w$ with $\mathsf{C}_p(0) \colonequals \#0$, $\mathsf{C}_p(1) \colonequals \#1$, $\mathsf{C}_w(0) \colonequals \#011$ and $\mathsf{C}_w(1) \colonequals \#001$. Now let $\vec{v} = (v_1, \dots, v_d) \in \{0, 1\}^d$. Then we can extend this to vectors:
\[
	\mathsf{C}_p(\vec{v}) \colonequals @\, \mathsf{C}_p(v_1)\, \mathsf{C}_p(v_2)\, \dots\, \mathsf{C}_p(v_d)\, \$
\]
and
\[
	\mathsf{C}_w(\vec{v}) \colonequals @\, \mathsf{C}_w(v_1)\, \mathsf{C}_w(v_2)\, \dots\, \mathsf{C}_w(v_d)\, \$
\]

Note that $\mathsf{C}_p(\vec{v})$ is a subsequence of $\mathsf{C}_w(\vec{v}')$ for any $\vec{v}, \vec{v}' \in \{0, 1\}^d$. This is because any two vectors can still make an orthogonal set with a third vector, so we cannot exclude any vectors preemptively.

We also define $\overline{\mathsf{C}}_p$ and $\overline{\mathsf{C}}_w$ as mirrored versions of these encodings (both for bits and vectors), where the order of the characters in the output is inverted. We can then use the original encoding for one part of the pattern and the text and the mirrored encoding for the other part.

Lastly, we only really need one of the vectors in $A$, but since we cannot know which one ahead of time, we need to encode all of them. Therefore, we define $w_0 \colonequals (\mathcal{C}_w(\vec{0}))^{n-1}$ and the mirrored version $\overline{w}_0$. These will serve as sections in $w$ where all but one of the vectors of $A$ can be embedded. We now have the following pattern and text:
\[
	p \colonequals \overline{\mathsf{C}}_p(\vec{a}_n)\, \dots\, \overline{\mathsf{C}}_p(\vec{a}_1)\, \S\, \mathsf{C}_p(\vec{a}_1)\, \dots\, \mathsf{C}_p(\vec{a}_n)
\]
and
\[
	w \colonequals \overline{w}_0\, \#\, \overline{\mathsf{C}}_w(\vec{c}_n)\, \dots\, \overline{\mathsf{C}}_w(\vec{c}_1)\, \#\, \overline{w}_0\, \S\, w_0\, \#\, \mathsf{C}_w(\vec{b}_1)\, \dots\, \mathsf{C}_w(\vec{b}_n)\, \#\, w_0
\]

Note that the $\S$ character only occurs once in each of the text and pattern and divides them into two parts each. For any subsequence, the $\S$ character of $p$ must be matched to the $\S$ character in $w$. Therefore, the left and right sides of $p$ will be embedded in the left and right sides of $w$ respectively.

From here on, we will call the encoding of a vector (in either $p$ or $w$) a block. Hereby, the $\overline{\mathsf{C}}_w(\vec{0})$ and $\mathsf{C}_w(\vec{0})$ blocks will be referred to as the \emph{zero-blocks}, while we call $\mathsf{C}_w(\vec{b}_j)$ and $\overline{\mathsf{C}}_w(\vec{c}_k)$ for $j, k \in [n]$ \emph{non-zero-blocks} (even though $\vec{b}_j$ or $\vec{c}_k$ could be zero vectors). The encoding of the $u$-th bit of that vector (with $u \in [d]$) will be called the \emph{$u$-th component} of the block. The $@$ character will be referred to as the \emph{start} of the block, even for the mirrored encoding, while the $\$$ character represents the \emph{end} of the block.

As can easily be verified, $p$ is always a subsequence of $w$. We will now incrementally add different sets of constraints to our currently empty set of constraints $\mathcal{C}$ such that an embedding of $p$ into $w$ that satisfies these constraints must fulfill certain properties and, ultimately, the existence of an embedding satisfying all these constraints will show that there are three orthogonal vectors in the set. In order to define the constraints, we will use $s_i$ and $\overline{s}_i$ for denoting the starting positions of $\mathsf{C}_p(\vec{a}_i)$ and $\overline{\mathsf{C}}_p(\vec{a}_i)$ for $i \in [n]$ (i.e., $p[s_i] = @ = p[\overline{s}_i]$).

First, we want to make sure that any block in the pattern gets embedded into only one block of the text and that any components of that block get embedded into their respective counterpart in the other block. This can be achieved by adding gap constraints as in \autoref{fig_gap_constraints}: For some block of $p$, the initial $@$ character of this block is embedded into an $@$ character initiating some block in $w$. Then the $\#$ character at the start of the first component should be embedded right next to the $@$ character, to ensure that it is embedded into the start of the first component of the block in $w$. This is achieved by the constraints
\begin{equation}
    \label{eq_constraint_gap_start}
    (s_i, s_i + 1, L(0)), (\overline{s}_i - 1, \overline{s}_i, L(0)) \qquad \text{for $i \in [n]$}
\end{equation}
The rest of the characters in the block can have at most 2 positions between them (note that $\{0, 1, 2\} = L(0) \cup  L(1) \cup L(2) )$ is a semilinear set):
\begin{equation}
    \label{eq_constraint_gap_middle}
	(s_i + v, s_i + v + 1, \{0, 1, 2\}), (\overline{s}_i - v - 1, \overline{s}_i - v, \{0, 1, 2\}) \qquad \text{for $i \in [n], v \in [2d]$}
\end{equation}
This ensures that the $\#$ characters in neighboring components are at most 6 characters apart and therefore must be mapped to $\#$ characters of adjacent components in $w$. At the same time, this still allows the digit in some component in $p$ to be embedded into any of the digits in the respective component in $w$ (see \autoref{fig_gap_constraints}). Finally, the $\$$ at the end of the block gets mapped to its counterpart in the other block. Because we added all these constraints in a mirrored manner to the blocks on the left side of $p$ this behaviour is ensured on both sides.

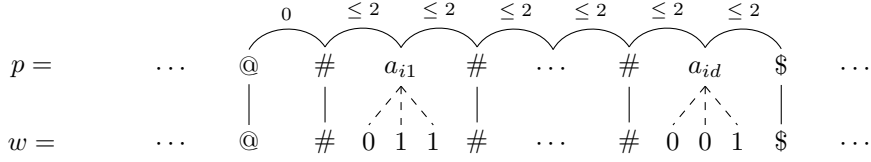
\begin{figure}[ht]
	\centering
	\begin{tikzpicture}[every node/.style={text height=.6em}]
		\node (P) {$p = {}$};
		\node (P1) [right= of P]{$\cdots$};
		\node (P2) [right of = P1]{$@$};
		\node (P3) [right of = P2]{$\#$};
		\node (P4) [right of = P3]{$a_{i1}$};
		\node (P5) [right of = P4]{$\#$};
		\node (P6) [right of = P5]{$\cdots$};
		\node (P7) [right of = P6]{$\#$};
		\node (P8) [right of = P7]{$a_{id}$};
		\node (P9) [right of = P8]{$\$$};
		\node (P10) [right of = P9]{$\cdots$};

		\node (W) [below of = P]{$w = {}$};
		\node (W1) [below of = P1]{$\cdots$};
		\node (W2) [below of = P2]{$@$};
		\node (W3) [below of = P3]{$\#$};
		\node (W4) [below of = P4]{$1$};
		\node (W4l) [left=0mm of W4]{$0$};
		\node (W4r) [right=0mm of W4]{$1$};
		\node (W5) [below of = P5]{$\#$};
		\node (W6) [below of = P6]{$\cdots$};
		\node (W7) [below of = P7]{$\#$};
		\node (W8) [below of = P8]{$0$};
		\node (W8l) [left=0mm of W8]{$0$};
		\node (W8r) [right=0mm of W8]{$1$};
		\node (W9) [below of = P9]{$\$$};
		\node (W10) [below of = P10]{$\cdots$};

		\draw (P2.north) to[bend left = 60] node[midway, above] {\scriptsize $0$} (P3.north);
		\draw (P3.north) to[bend left = 60] node[midway, above] {\scriptsize $\leq 2$} (P4.north);
		\draw (P4.north) to[bend left = 60] node[midway, above] {\scriptsize $\leq 2$} (P5.north);
		\draw (P5.north) to[bend left = 60] node[midway, above] {\scriptsize $\leq 2$} (P6.north);
		\draw (P6.north) to[bend left = 60] node[midway, above] {\scriptsize $\leq 2$} (P7.north);
		\draw (P7.north) to[bend left = 60] node[midway, above] {\scriptsize $\leq 2$} (P8.north);
		\draw (P8.north) to[bend left = 60] node[midway, above] {\scriptsize $\leq 2$} (P9.north);

		\draw (P2.south) to (W2.north);
		\draw (P3.south) to (W3.north);
		\draw[dashed] (P4.south) to (W4.north);
		\draw[dashed] (P4.south) to (W4l.north);
		\draw[dashed] (P4.south) to (W4r.north);
		\draw (P5.south) to (W5.north);
		\draw (P7.south) to (W7.north);
		\draw[dashed] (P8.south) to (W8.north);
		\draw[dashed] (P8.south) to (W8l.north);
		\draw[dashed] (P8.south) to (W8r.north);
		\draw (P9.south) to (W9.north);
	\end{tikzpicture}
	\caption{Gap constraints on $\mathsf{C}_p(\vec{a}_i)$. The constraint depicted on each edge concerns the string which can be found between the embeddings of the two positions that edge connects.}\label{fig_gap_constraints}
\end{figure}

Now, to satisfy our current set of constraints, blocks need to be embedded in other blocks, but while there are $n$ blocks on either side of $p$ (where the left side is to the left of $\S$, and the right side to the right of $\S$), each side of $w$ (defined similarly w.r.t. the position of $\S$) has $2n - 2$ zero-blocks. Therefore, right now all the vectors in $A$ can be matched to zero vectors and we cannot gain any information about the vectors in $B$ or $C$. To fix this, we connect adjacent blocks on the right side of $p$ with a new gap constraint: The embedded positions of the ending of one block and the start of the next block can then either be right next to each other, or they have to be an even distance apart (i.e., have a string of odd length between them): 
\begin{equation}
    \label{eq_constraint_gap_end}
    (s_i - 1, s_i, L(0) \cup L(1;2)\}) \qquad \text{for $i \in [2, n]$}
\end{equation}
To see why this works, let us consider starting and ending positions of blocks in $w$. Because the blocks are of even length, blocks directly beside each other either both have even starting positions or both have odd starting positions. Because we introduced extra $\#$ characters at some positions in $w$, this means that all the zero-blocks start at even positions (remember that we defined the start of the mirrored block $\overline{\mathsf{C}}_w(\vec{0})$ to be at the position of the $@$ character, therefore on the right end of the block), while all the non-zero-blocks start at odd positions. Also, if a block starts at an odd position, it ends on an even position and vice versa.
Because we can just barely not fit all the blocks on the right side of $p$ into the same $w_0$ section, since that only contains $n-1$ blocks, at least one of them must be in a different section. Thus, we get two blocks that are not embedded directly besides each other, which means that the ending of the first block and the start of the second block must be either both on even or both on odd positions. Therefore, not both these blocks can be embedded into a $w_0$ section. We thus get $i^*, j^* \in [n]$, such that $\mathsf{C}_p(\vec{a}_{i^*})$ gets embedded in $\mathsf{C}_w(\vec{b}_{j^*})$.

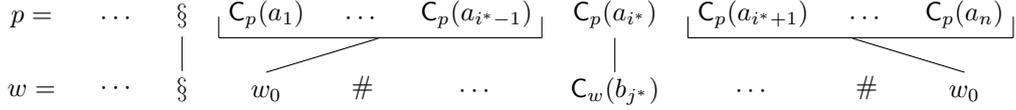
\begin{figure}[ht]
	\centering
	\begin{tikzpicture}[every node/.style={text height=.6em}]
		\node (P) {$p = {}$};
		\node (P1) [right=.8em of P]{$\cdots$};
		\node (P2) [right=.8em of P1]{$\S$};
		\node (P3) [right=.8em of P2]{$\mathsf{C}_p(a_1)$};
		\node (P4) [right=.8em of P3]{$\cdots$};
		\node (P5) [right=.8em of P4]{$\mathsf{C}_p(a_{i^*-1})$};
        \node (P6) [right=.8em of P5]{$\mathsf{C}_p(a_{i^*})$};
		\node (P7) [right=.8em of P6]{$\mathsf{C}_p(a_{i^*+1})$};
		\node (P8) [right=.8em of P7]{$\cdots$};
		\node (P9) [right=.8em of P8]{$\mathsf{C}_p(a_n)$};

		\node (W) [below=1.3em of P]{$w = {}$};
		\node (W1) [below=1.3em of P1]{$\cdots$};
		\node (W2) [below=1.3em of P2]{$\S$};
		\node (W3) [below=1.3em of P3]{$w_0$};
		\node (W4) [below=1.3em of P4]{$\#$};
		\node (W5) [below=1.3em of P5]{$\cdots$};
        \node (W6) [below=1.3em of P6]{$\mathsf{C}_w(b_{j^*})$};
		\node (W7) [below=1.3em of P7]{$\cdots$};
		\node (W8) [below=1.3em of P8]{$\#$};
		\node (W9) [below=1.3em of P9]{$w_0$};

		\draw (P2.south) to (W2.north);
		\draw (P3.west) to (P3.south west);
		\draw (P5.east) to (P5.south east);
		\draw (P3.south west) to node[midway] (Pl) [inner sep=0pt] {} (P5.south east);
		\draw (Pl) to (W3.north);
		\draw (P6.south) to (W6.north);
		\draw (P7.west) to (P7.south west);
		\draw (P9.east) to (P9.south east);
		\draw (P7.south west) to node[midway] (Pr) [inner sep=0pt] {} (P9.south east);
		\draw (Pr) to (W9.north);
	\end{tikzpicture}
	\caption{Embedding the blocks on the right side}\label{fig_gap_constraints2}
\end{figure}

Next we want to make sure that on the left side $\overline{\mathsf{C}}_p(\vec{a}_{i^*})$ is embedded into $\overline{\mathsf{C}}_w(\vec{c}_{k^*})$ for some $k^* \in [n]$. Since we saw earlier that both $\overline{\mathsf{C}}_w(\vec{c}_{k^*})$ and $\mathsf{C}_w(\vec{b}_{j^*})$ start on an odd position (while all zero-blocks start on even positions), we can require that for any $i \in [n]$, the embedded starting positions of $\overline{\mathsf{C}}_p(\vec{a}_i)$ and $\mathsf{C}_p(\vec{a}_i)$ have to be an even number of positions apart (i.e., have an odd number of symbols between them):
\begin{equation}
    \label{eq_constraint_bridge}
    (\overline{s}_i, s_i, L(1;2)) \qquad \text{for $i \in [n]$}
\end{equation}
This ensures that, for an arbitrary $i \in [n]$, $\overline{\mathsf{C}}_p(\vec{a}_i)$ and $\mathsf{C}_p(\vec{a}_i)$ either both get embedded in zero-blocks or both get embedded in non-zero-blocks and as such $\overline{\mathsf{C}}_p(\vec{a}_{i^*})$ is embedded into some block $\overline{\mathsf{C}}_w(\vec{c}_{k^*})$.

Having selected three vectors, $\vec{a}_{i^*}$, $\vec{b}_{j^*}$ and $\vec{c}_{k^*}$, it remains to be seen how to ensure that these three vectors are indeed orthogonal to each other. The three vectors are orthogonal, if and only if for all $u \in [d]$, at least one of $a_{i^*u}$ or $b_{j^*u}$ or $c_{k^*u}$ is zero. Conversely, the three vectors are not orthogonal, if and only if there exists $u \in [d]$, such that $a_{i^*u} = b_{j^*u} = c_{k^*u} = 1$. In this case, $\overline{\mathsf{C}}_w(c_{k^*u}) = 100\#$ and $\mathsf{C}_w(b_{j^*u}) = \#001$ only contain a single ``1'' character. The distance of these characters to the start of their respective blocks is the same and the length of the factor found between the starting positions of the two blocks is guaranteed to be odd due to constraint \cref{eq_constraint_bridge}. Therefore, the factor between these two characters is also odd and we can prevent such an embedding to occur by adding a constraint between the ``1'' in $\overline{\mathsf{C}}_p(a_{iu})$ and the ``1'' in $\mathsf{C}_p(a_{iu})$ that requires the distance between their embeddings to be even (see \autoref{fig_constraint_bridge}):
\begin{equation}
	\label{eq_constraint_ortho}
    (\overline{s}_i - 4u - 1, s_i + 4u + 1, L(0; 2)) \qquad \text{for $i \in [n], u \in [d]$ with $a_{iu} = 1$}
\end{equation}
Note that this constraint is only added if $a_{iu} = 1$, because otherwise the case $a_{iu} = b_{ju} = c_{ku} = 0$ would also be prevented.

\begin{figure}[ht]
	\centering
	\begin{tikzpicture}[every node/.style={text height=.6em}]
        \node (P) {$p = {}$};
		\node (lP7) [right=2mm of P]{$\cdots$};
		\node (lP6) [right=2mm of lP7]{$1$};
		\node (lP5) [right=2mm of lP6]{$\#$};
		\node (lP4) [right=2mm of lP5]{$0$};
		\node (lP3) [right=2mm of lP4]{$\#$};
		\node (lP2) [right=2mm of lP3]{$@$};
		\node (lP1) [right=2mm of lP2]{$\cdots$};
		\node (P0)  [right=2mm of lP1]{$\S$};
		\node (rP1) [right=2mm of P0]{$\cdots$};
		\node (rP2) [right=2mm of rP1]{$@$};
		\node (rP3) [right=2mm of rP2]{$\#$};
		\node (rP4) [right=2mm of rP3]{$0$};
		\node (rP5) [right=2mm of rP4]{$\#$};
		\node (rP6) [right=2mm of rP5]{$1$};
		\node (rP7) [right=2mm of rP6]{$\cdots$};

        \node (W) [below of = P]{$w = {}$};
		\node (lW7) [below of = lP7]{$\cdots$};
		\node (lW6) [below of = lP6, inner xsep=0pt]{$0$};
        \node (lW6l) [left=0mm of lW6, inner xsep=0pt] {$1$};
        \node (lW6r) [right=0mm of lW6, inner xsep=0pt] {$0$};
		\node (lW5) [below of = lP5]{$\#$};
		\node (lW4) [below of = lP4, inner xsep=0pt]{$1$};
        \node (lW4l) [left=0mm of lW4, inner xsep=0pt] {$1$};
        \node (lW4r) [right=0mm of lW4, inner xsep=0pt] {$0$};
		\node (lW3) [below of = lP3]{$\#$};
		\node (lW2) [below of = lP2]{$@$};
		\node (lW1) [below of = lP1]{$\cdots$};
		\node (W0)  [below of = P0]{$\S$};
        \node (rW1) [below of = rP1]{$\cdots$};
		\node (rW2) [below of = rP2]{$@$};
		\node (rW3) [below of = rP3]{$\#$};
		\node (rW4) [below of = rP4, inner xsep=0pt]{$0$};
        \node (rW4l) [left=0mm of rW4, inner xsep=0pt] {$0$};
        \node (rW4r) [right=0mm of rW4, inner xsep=0pt] {$1$};
		\node (rW5) [below of = rP5]{$\#$};
		\node (rW6) [below of = rP6, inner xsep=0pt]{$1$};
        \node (rW6l) [left=0mm of rW6, inner xsep=0pt] {$0$};
        \node (rW6r) [right=0mm of rW6, inner xsep=0pt] {$1$};
		\node (rW7) [below of = rP7]{$\cdots$};

		\draw (lP2.north) to[bend left = 60] node[midway, above] {odd} (rP2.north);
		\draw (lP6.north) to[bend left = 60] node[midway, above] {even} (rP6.north);

		\draw (P0.south) -- (W0.north);
		\draw (lP2.south) -- (lW2.north);
		\draw (rP2.south) -- (rW2.north);
        \draw (lP3.south) -- (lW3.north);
        \draw (rP3.south) -- (rW3.north);
        \draw (lP4.south) -- (lW4r.north);
        \draw[dashed] (rP4.south) -- (rW4.north);
        \draw[dashed] (rP4.south) -- (rW4l.north);
        \draw (lP5.south) -- (lW5.north);
        \draw (rP5.south) -- (rW5.north);
        \draw (lP6.south) -- (lW6l.north);
        \draw (rP6.south) -- (rW6.north);
	\end{tikzpicture}
    \caption{Constraints between $\overline{\mathsf{C}}_p(\vec{a}_i)$ and $\mathsf{C}_p(\vec{a}_i)$ (exemplified)}\label{fig_constraint_bridge}
\end{figure}
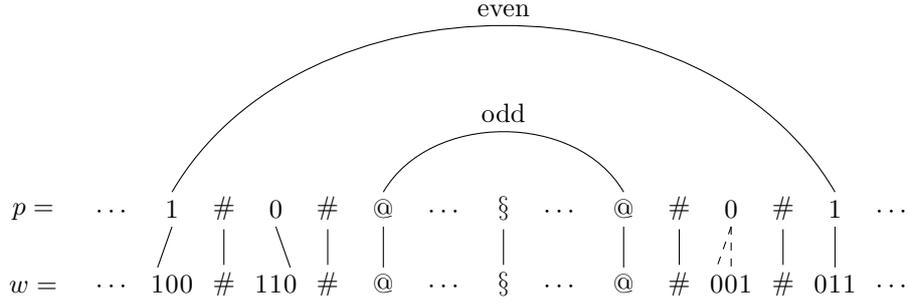

With $\mathcal{C}$ being the set of constraints \cref{eq_constraint_gap_start,eq_constraint_gap_middle,eq_constraint_gap_end,eq_constraint_bridge,eq_constraint_ortho}, we can say that if $p$ is a $\mathcal{C}$-subsequence of $w$, then the answer to the $3$-$\OV$ problem must be positive. Note that none of the constraints intersect: The gap constraints, i.e.\ \cref{eq_constraint_gap_start,eq_constraint_gap_middle,eq_constraint_gap_end}, cannot intersect with any other constraint, because their respective intervals are singletons and therefore any other interval either contains it or does not intersect with it. The other constraints are connecting the left and right part of $p$. Here, for any constraint, both positions of the constraint (in $p$) have the same distance to the center of $p$. Therefore, any two of these constraints have to be comparable w.r.t.\ $\sqsubset$ and thus, cannot intersect.

We still need to show that whenever there is a triple $\vec{a}_{i^*}$, $\vec{b}_{j^*}$ and $\vec{c}_{k^*}$ of orthogonal vectors, $p$ must be a $\mathcal{C}$-subsequence of $w$. Let us consider the right side of $p$: Since the $w_0$ sections each contain $n-1$ zero-blocks, we can embed the $i^*-1 \leq n-1$ blocks $\mathsf{C}_p(\vec{a_1}), \dots, \mathsf{C}_p(\vec{a}_{i^*-1})$ into the left $w_0$ section and the $n - i^* \leq n - 1$ blocks $\mathsf{C}_p(\vec{a}_{i^*+1}), \dots, \mathsf{C}_p(\vec{a}_n)$ into the right $w_0$ section. Obviously, we embed $\mathsf{C}_p(\vec{a}_{i^*})$ into $\mathsf{C}_w(\vec{b}_{j^*}$) (see \autoref{fig_gap_constraints2}). Then the constraints \cref{eq_constraint_gap_start,eq_constraint_gap_middle,eq_constraint_gap_end} are satisfied. Also, if we embed the left side of $p$ analogously with $\overline{\mathsf{C}}_p(\vec{a}_{i^*})$ being embedded into $\overline{\mathsf{C}}_w(\vec{c}_{k^*})$, then for all $i \in [n]$ the embedded starting positions of $\overline{\mathsf{C}}_p(\vec{a}_{i})$ and $\mathsf{C}_p(\vec{a}_{i})$ are either both even (for $i = i^*$) or both odd (for $i \neq i^*$) and therefore always an even distance apart, satisfying \eqref{eq_constraint_bridge}. The last set of constraints can also be satisfied: Because $\vec{a}_{i^*}$, $\vec{b}_{j^*}$ and $\vec{c}_{k^*}$ are orthogonal, for any $u \in [d]$, either one of $a_{i^*u}$, $b_{j^*u}$ or $c_{k^*u}$ must be zero. If $a_{i^*u} = 0$, then there is no constraint that needs to be satisfied. If $b_{j^*u}$ or $c_{k^*u}$ are zero, then $\overline{\mathsf{C}}_w(c_{k^*u})$ or $\mathsf{C}_w(b_{j^*u})$ respectively contain two adjacent ``1'' characters, which means that it is always possible to embed the ``1'' characters in $\overline{\mathsf{C}}_p(a_{i^*u})$ and $\mathsf{C}_p(a_{i^*u})$ at an even distance from each other (i.e., they have an odd factor in between them; for an example, see \autoref{fig_constraint_bridge}).
For all $i \neq i^*$, the triple $(\vec{a}_{i}, \vec{0}, \vec{0})$ is orthogonal and constraints \eqref{eq_constraint_ortho} can be satisfied as well.

Therefore, the answer to the $3$-$\OV$ problem instance is true, \emph{if and only if} $p$ is a $\mathcal{C}$-subsequence of $w$. With this, we can show the result stated in Theorem \ref{thm:NonIntersectLB}: determining whether $p$ is a $\mathcal{C}$-subsequence of $w$ cannot be done in $\mathcal{O}(|w|^g|\mathcal{C}|^h)$ time with $g + h < 3$, unless \textsf{SETH} (\autoref{hyp_SETH}) fails.

Indeed, assume that there was such an algorithm. Then we could solve $3$-$\OV$ with three sets of $n$ $d$-dimensional vectors by converting these into $p$, $w$ and $\mathcal{C}$ as shown above (in $\mathcal{O}(nd)$ time). We then have $|w|, |\mathcal{C}| \in \mathcal{O}(nd)$ and can figure out whether $p$ is a $\mathcal{C}$-subsequence of $w$ in $\mathcal{O}(n^{g+h}d^{g+h})$ time, which in turn gives the solution to the $3$-$\OV$ problem instance. In summary, we would be able to solve $3$-$\OV$ in $\mathcal{O}(n^{g + h}\operatorname{poly}(d))$ time (with $g + h < 3$) and, by \autoref{lem_k_OV}, the Strong Exponential Time Hypothesis would have to be false.

This concludes the proof of Theorem \ref{thm:NonIntersectLB} for the case of $\semi$-constraints. However, all constraints in the reduction can be also stated as DFAs with a constant number of states, so the result canonically holds for $\reg$-constraints.

\end{document}